\documentclass[lettersize,journal]{IEEEtran}
\usepackage{amsmath,amsfonts}
\usepackage{algorithmic}
\usepackage{algorithm}
\usepackage{array}
\usepackage[caption=false,font=footnotesize,labelfont=rm,textfont=rm]{subfig}
\usepackage{textcomp}
\usepackage{stfloats}
\usepackage{url}
\usepackage{verbatim}
\usepackage{graphicx}
\usepackage{cite}
\usepackage{color}
\usepackage{amssymb}
\usepackage{amsthm}

\newtheorem{theorem}{Theorem}
\newtheorem{Lemma}{Lemma}
\newtheorem{corollary}{Corollary}
\newtheorem{proposition}{Proposition}
\newtheorem{definition}{Definition}

\begin{document}

\title{Optimizing Peak Age of Information in MEC Systems: Computing Preemption and Non-preemption}
\author{Jianhang~Zhu and Jie~Gong,~\IEEEmembership{Member,~IEEE}
\IEEEcompsocitemizethanks
{\IEEEcompsocthanksitem J. Zhu and J. Gong are with the School of Computer Science and Engineering, and the Guangdong Key Laboratory of Information Security Technology, Sun Yat-sen University, Guangzhou 510006, China. Emails: zhujh26@mail2.sysu.edu.cn, gongj26@mail.sysu.edu.cn.
}
}

\maketitle
\begin{abstract}
  The freshness of information in real-time monitoring systems has received increasing attention, with Age of Information (AoI) emerging as a novel metric for measuring information freshness. In many applications, update packets need to be computed before being delivered to a destination. Mobile edge computing (MEC) is a promising approach for efficiently accomplishing the computing process, where the transmission process and computation process are coupled, jointly affecting freshness. In this paper, we aim to minimize the average peak AoI (PAoI) in an MEC system. We consider the generate-at-will source model and study when to generate a new update in two edge server setups: 1) computing preemption, where the packet in the computing process will be preempted by the newly arrived one, and 2) non-preemption, where the newly arrived packet will wait in the queue until the current one completes computing. We prove that the fixed threshold policy is optimal in a non-preemptive system for arbitrary transmission time and computation time distributions. In a preemptive system, we show that the transmission-aware threshold policy is optimal when the computing time follows an exponential distribution. Our numerical simulation results not only validate the theoretical findings but also demonstrate that: 1) in our problem, preemptive systems are not always superior to non-preemptive systems, even with exponential distribution, and 2) as the ratio of the mean transmission time to the mean computation time increases, the optimal threshold increases in preemptive systems but decreases in non-preemptive systems.
\end{abstract}


\section{Introduction}
Recently, there has been a trend in mobile computing towards shifting cloud functions to network edges, such as base stations and access points, in order to utilize the vast amount of idle computation power and storage space for computation-intensive and latency-critical tasks of mobile devices. This trend is known as Mobile Edge Computing (MEC) \cite{mao2017survey}. MEC provides low-latency services, and combined with the extensive data collection capabilities of the Internet of Things (IoT) \cite{IoT2015}, it enables various real-time applications such as remote monitoring and control, phase packet update in smart grids, and environment monitoring for autonomous driving. The performance of these systems is closely tied to the freshness of the information they provide. The concept of \emph{Age of Information} (AoI), defined as the time elapsed since the generation of the last received update, was introduced in \cite{kaul2012real} to quantify this freshness. In MEC systems, in addition to transmission latency and update frequency, the impact of computation on AoI should also be taken into account.

\begin{figure}[!t]
  \centering
  \includegraphics[width=0.9\linewidth]{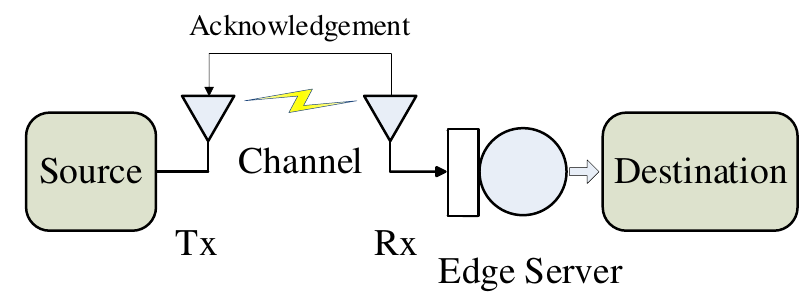}
  \caption{Status update system with MEC.} \label{fig:system}
\end{figure}

To address this issue, we study a state update system incorporating MEC, as illustrated in Fig.~\ref{fig:system}, where a source node generates update packets and transmits them to an edge server via a channel. The server then computes the packets and sends the results to a destination. There is a queue at the server to buffer the received packets when necessary. Both transmission time and computation time are assumed to be random and unknown apriori. Computation affects AoI performance in two ways. First, computation delay at the edge server directly increases the age of updates. Second, updates may wait in the server queue due to limited computing resources at the edge server and hence, become stale. When computation time is taken into account, the source should choose to generate and transmit a new update either before or after current update's computation is complete. The former option utilizes the channel and the edge server simultaneously, which may benefit AoI performance, but may also result in new data becoming stale in the queue. The latter can completely avoid queueing at the edge server, but increases the inter-generation time. Thus, the status update process must be carefully designed to minimize AoI by avoiding queueing in the server while reducing the inter-generation time.

As mentioned before, transmitting a new update before the completion of the current update's computation can potentially reduce the AoI, as it better utilizes the transmission resource of the channel and the computing resource of the edge server. However, we inevitably encounter queueing at the edge server which is harmful to AoI performance. To tackle the problem, this paper investigates two server configurations: 1) computing preemption where the packet in computing will be preempted by the newly arrived one and 2) non-preemption where the newly arrived packet will wait in queue until the current one completes computing. As new update is fresher, systems with preemptive server can typically achieve better performance. In \cite{kaul2012status}, it is shown that an LCFS queue with preemption achieves a lower average age compared to the model without preemption. However, whether preemptive server outperforms non-preemptive server in MEC case is not clear, which will be investigated in our work.


\subsection{Related Work}

The Age of Information has become an important metric in various status update systems, leading to studies aiming at minimizing it. These studies can be classified into two categories based on the type of source. The first category focuses on uncontrollable sources, as seen in works such as \cite{kaul2012real,yate2012real,bt2015age,huang2015opt,bed2019the,yates2019theage}. In these works, minimizing AoI typically involves optimizing the service rate and packet management strategy in the queues. For instance, \cite{kaul2012real} analyzed the average AoI in the First-Come-First-Serve (FCFS) system for different types of queues and found that there exists an optimal offered load for each queue to minimize the average AoI.  In \cite{yates2019theage}, a new simplified technique for evaluating AoI in finite-state continuous-time queuing systems was derived for the multi-source LCFS system with Poisson arrivals and exponential service time. The second category considers controllable sources, also known as the generate-at-will source model, as seen in \cite{sun2017update,arafa2020age,gu2021opt,cham2021min,gong2022sleep,arafa2021timely}. For example, \cite{sun2017update} considered a generate-at-will source model where updates can be generated at any time according to a scheduling policy and proved that a threshold-based policy is optimal when considering independent and identically distributed transmission times. However, these works did not consider the impact of computation on AoI.

The impact of computing on AoI has garnered increasing attention~\cite{alabbasi2018joint, arafa2019timely, zou2021opt, kuang2019age, gong2019reducing, zhong2019age, song2019age, li2021age, zhu2022online}, as the information contained in a status update packet is often not revealed until it has been processed. Ref.\cite{alabbasi2018joint} explored the impact of computing on AoI by scheduling computing tasks in the central cloud. The scheduling policy for update cloud computing, ignoring transmission time, was studied in~\cite{arafa2019timely}. In~\cite{zou2021opt}, the trade-off between computing and transmission was analyzed, where each packet was pre-processed before being transmitted. In~\cite{kuang2019age} and~\cite{gong2019reducing}, the average AoI with exponential transmission time and service time was analyzed for the single-user case when MEC was considered. However, these works did not address the fundamental question of what the optimal scheduling policy is to achieve the minimum AoI in an MEC system that considers both transmission and computation times. Recently, the authors in \cite{cham2021min} investigated the optimal scheduling policy in a single-source single-server system with a fixed request delay. The requests initiated by the monitor arrive at the source after a fixed delay, and then the source generates an update. Transmission time is not taken into account, so there is no challenge posed by the transmission-computation coupling problem. In our research, we consider both transmission and computation time to address this fundamental question.

Peak AoI is a variant of AoI that reflects the average of the maximum AoI values, introduced in \cite{costa2014age}. In recent years, it has been further investigated in the literature \cite{costa2016on, abd2019aPAoI, ino2019agen, Najm2018status, Bedewy2021Low-Power, Yang2022Edge}. In \cite{abd2019aPAoI}, the role of unmanned aerial vehicles (UAV) as mobile relays to minimize the average PAoI for a source-destination pair is explored. An optimization problem is formulated to jointly optimize the UAV's flight trajectory as well as energy and service time allocations for packet transmissions. Peak AoI is not an age penalty functional~\cite{yates2021age}, however, the work in \cite{ino2019agen} considers AoI in single-server queues with different service disciplines. The results indicate that the stationary distribution of AoI can be expressed in terms of the stationary distributions of the system delay and PAoI, with the average PAoI serving as a fundamental characterization of the age process. In \cite{Najm2018status}, the expression for the PAoI in each source of a multi-source M/G/1/1 queue is considered. Ref.~\cite{Bedewy2021Low-Power} investigates a network of battery-constrained nodes, aiming to design a sleep-wake strategy that minimizes the weighted average PAoI while satisfying energy constraints. Ref.~\cite{Yang2022Edge} discusses the real-time performance of a cache-enabled network with Peak AoI as a freshness metric. They propose a random caching framework and demonstrate its superiority over content-based and uniform caching strategies in minimizing the PAoI. Nevertheless, the PAoI performance in MEC systems remains unexplored.

\subsection{Main Results}

In this paper, we aim to minimize the average peak AoI in MEC systems while considering transmission delay and computation time. The motivation of choosing this metric is similar to~\cite{costa2014age}\cite{costa2016on}. On the one hand, the average PAoI is more meaningful for scenarios that require reducing AoI violation probability, rather than the average AoI. On the other hand, optimizing average AoI in systems with generate-at-will sources is usually challenging~\cite{sun2017update}, and only heuristic strategies can be proposed in continuous-time systems~\cite{zhu2022online}. Even in a cloud computing system without considering transmission time, minimizing average AoI for general service time is still an open problem \cite{arafa2019timely}. Typically, PAoI minimization leads to well-structured solutions~\cite{cham2021min}. We address the problem of minimizing the average PAoI in both preemptive and non-preemptive systems, where the transmission time assumed i.i.d. and so does computation time. Specifically, we consider a generate-at-will source, which can generate updates at any time as long as the channel is available\footnote{In general systems, the source has a queue to store new updates when the channel is unavailable. However, generating a new update when the channel is unavailable is obviously sub-optimal when considering the freshness of information. Therefore, in this paper, we consider the source has no queue and only generates updates when the channel is available.}, and seek the optimal scheduling policy to minimize the average PAoI. The key contributions of this paper are:

\begin{itemize}
  \item We formulate the problem of minimizing the average PAoI, which is compatible with both systems. The feasible policy space consists of all causal policies, where the control decisions depend on the history and current information of the system. 
  \item In the non-preemptive system, we first prove that the optimal policy is a continuous working policy (Lemma~\ref{lem:maxZ}). Then, we prove that a randomized threshold policy is optimal (Theorem~\ref{The:WOP_RT}). Subsequently, we prove that a fixed threshold policy is optimal (Theorem~\ref{The:WOP_FT}), and derive the conditions that the optimal threshold must satisfy (Proposition~\ref{pro:fis}). Finally, we discuss two special cases. With exponentially distributed computation time, the optimal threshold is given in closed-form. With exponentially distributed transmission time, a piece-wise bisection method is proposed to find the optimal threshold.
  \item In the preemptive system, we first show that the continuous working policy is still optimal (Corollary~\ref{cor:maxZ_wp}). We then prove that a stationary deterministic policy is optimal in this case (Theorem~\ref{The:WP_SD}). Thus, the original problem can be reformulated as a non-convex functional optimization problem. We prove that the transmission-aware threshold policy is optimal when the computation time follows an exponential distribution (Theorem~\ref{The:optiamlC}). We also consider the fixed threshold policy, and propose an iterative algorithm to find the optimal threshold.
  \item In our numerical study, we evaluate the performance of the proposed policies using exponentially distributed computation time and two transmission time distributions, namely exponential and Pareto, to investigate the minimum PAoI and optimal thresholds. Our experimental results not only confirm the analytical findings, but also demonstrate that: 1) preemptive systems are not always superior to non-preemptive systems even when the memoryless exponential distribution is assumed, and 2) the optimal threshold increases in preemptive systems but decreases in non-preemptive systems as the ratio of mean transmission time to mean computation time increases.
\end{itemize}


The remainder of the paper is structured as follows. In Section \ref{sec:sys}, we introduce the MEC system model and formulate the problem of minimizing the average PAoI. Section \ref{sec:nonpreemption} focuses on the optimal policy in the non-preemptive system, while Section \ref{sec:preemption} explores the minimizing problem in the preemptive system. In Section \ref{sec:num}, we present numerical results, and in Section \ref{sec:con}, we draw conclusions based on our findings.
  
\section{System Model And Problem Formulation}\label{sec:sys}
\subsection{System Model}\label{subsec:systemmodel}

We consider a status update system as shown in Fig.~\ref{fig:system}, where a source generates real-time status packets and sends them to an edge server though a channel. The source generates and submits packets to the channel at times $S_0,S_1,\dots$. Each packet arrives at the edge server after a random transmission time, and then is served for a random computation time, and finally is delivered to the destination. Packets arriving while the edge server is busy can be stored in a queue. The source is aware of the idle/busy state of the channel and the edge server through the feedback link between the source and the edge server. For example, the edge server sends 0 when a new packet arrives and 1 when the packet finishes computing. We assume the transmission time of the feedback signal is negligible.

Intuitively, discarding updates in the queue when a new update arrives at the server can lead to a lower AoI. However, this makes a policy of sending as much as possible potentially optimal, which results in a large number of packets being dropped. Therefore, in this paper, we consider an FCFS queue, and the source can only submit new data when the queue is empty. Thus, a unit-sized queue is considered in our model.

In order to fully utilize transmission resources and hence potentially reduce the AoI, the source can submit a new packet when the channel is idle and the queue is empty, even if the edge server is busy. Unfortunately, the server may still be occupied when a new packet arrives, since the transmission time is random. In this case, we consider two systems: one with preemption and one without preemption as follows:
\begin{enumerate}
  \item \textit{System without service preemption.} In the non-preemptive service case, newly arrived packets are stored in the unit-sized queue if the edge server is busy, and the packets wait until the server becomes idle. During the waiting period, the source does not generate new packets. Thus, all packets are successfully delivered.
  \item \textit{System with service preemption.} In the preemptive service case, the edge server always serves the newly arrived packet and discards the old one. Therefore, the queue is always empty. No packets need to wait in the queue, but some packets will be dropped due to preemption.
\end{enumerate}

Suppose the update $i$ is generated and submitted at time $S_i$, its transmission time is $T_i$, and its computation time is $C_i$. Assume the transmission times of the packets are i.i.d with a positive finite mean $0<E[T]<\infty$, and make the same assumptions for the computation times. Let $F_{T}(\cdot),f_{T}(\cdot)$ and $F_{C}(\cdot),f_{C}(\cdot)$ denote the cumulative distribution function and probability density function of the transmission time and computation time, respectively. Since both $T$ and $C$ are positive, we have $f_T(t) = f_C(t) = 0 $ for all $t \le 0$. Denote $W_i$ as the waiting time of packet $i$ in the queue. Hence, the packet $i$ is delivered at time 
\begin{align}
  D_i=S_i+T_i+W_i+C_i\label{equ:dk}
\end{align}
or is possibly dropped (in preemptive service case). We assign $D_i=\infty$, if packet $i$ is dropped due to preemption. We assume that packet $0$ is submitted at time $S_0 =-T_0-C_0$ and is delivered at $D_0=0$. The inter-generation time between packet $i+1$ and packet $i$, denoted by $Z_i$, is given by 
\begin{align}
  Z_i=S_{i+1}-S_i.\label{equ:zk}
\end{align}
Since the source can only submit packets when the channel is idle and the queue is empty, we have 
\begin{align}
  Z_i\ge T_i+W_i.
\end{align}
A scheduling policy can be written as $\pi=\{Z_i,i\ge 1\}$.

At time $t$, the AoI at the destination, denote by $\Delta(t)$, is given by
\begin{equation}
\Delta(t)=t-\mathop{\max}_{i\in \mathbb{N}}\{S_i|D_i\le t\}.\label{def:age}
\end{equation}
The initial AoI is $\Delta(0)=T_0+C_0$. Let $k$ denote the $k$-th successfully delivered packet. To calculate the PAoI, we denote: (i) $i_k$ to be the transmission index of the $k$-th successfully delivered packet, (ii) 
\begin{align}
  X_k=S_{i_{k}}-S_{i_{k-1}}\label{equ:xk}
\end{align}
to be the inter-generation time between two consecutive successfully delivered packets, (iii) 
\begin{align}
  A_k=D_{i_{k}}-S_{i_k}\label{equ:Ak}
\end{align}
to be the system time, or the time spent by the $k$-th successful packet in the channel and the edge server. Then, the PAoI can be calculated by 
\begin{align}
  P_k=X_k+A_k.\label{equ:Pk}
\end{align}

The above definitions are compatible for both preemptive and non-preemptive systems. In the system without service preemption, all updates are successfully delivered, thus, $i_k=k$ for all $k=1,2,\cdots$. The $k$-th packet's waiting time $W_k=\max\{0,D_{k-1}-S_{k}-T_{k}\}$. In the system with service preemption, packets never enter the queue and hence $W_i=0$ for all $i=1,2,\cdots$. The evolution
of the instantaneous age for these two scenarios is given in
Fig.~\ref{fig:Age_curve}. 

\textbf{Remark} Although PAoI is a more manageable age metric compared to AoI, it remains challenging in the context of joint optimization of transmission and computation. In contrast to treating transmission and computation processes as a single process \cite{sun2017update}, joint transmission-computation allows for overlapping system time of packets, i.e., transmitting new data while old data is being computed. On the one hand, this may reduce the inter-generation time, thus lowering the average peak AoI. On the other hand, it may lead to competition for edge server resources, increasing the system time and consequently raising the average peak age. Simply applying strategies from \cite{sun2017update} to this situation, that is, treating transmission and computation as a unified process, would overlook the advantages of joint transmission-computation optimization. Therefore, determining the generation time to achieve the joint optimization of average PAoI during busy edge server periods is our main challenge. In the non-preemptive system, this challenge is reflected in balancing reduced inter-generation time with increased queue time. In the preemptive system, this challenge manifests in balancing inter-generation time with the probability of successful delivery.

\begin{figure*}[!t]
  \begin{minipage}{1\linewidth}
      \centering
      \subfloat[Age of information in the non-preemptive system]{\includegraphics[width=.5\linewidth]{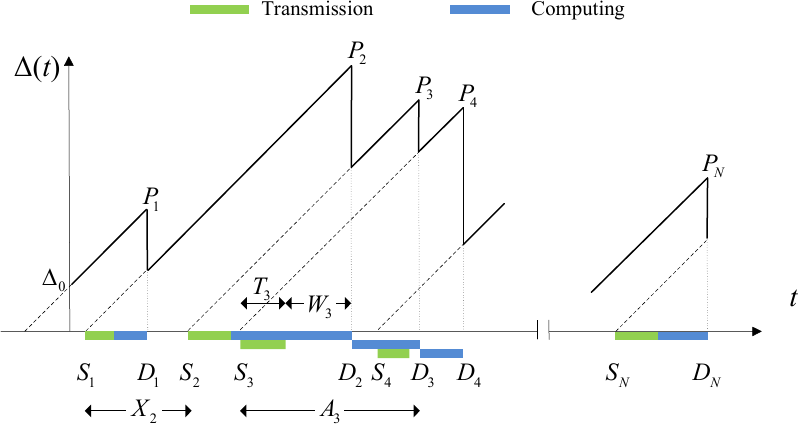}}
      \subfloat[Age of information in the  preemptive system]{\includegraphics[width=.5\linewidth]{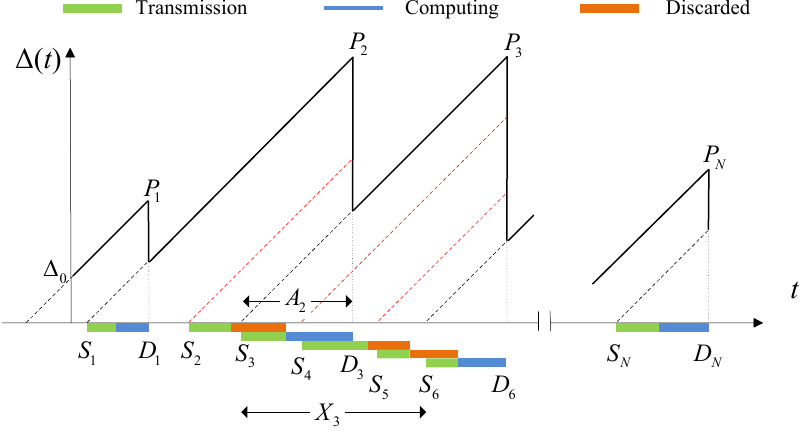}}
  \end{minipage}
  \caption{Age curve in two systems.}
  \label{fig:Age_curve}
\end{figure*}

\subsection{Problem Formulation}

Under a given scheduling policy $\pi$, the average PAoI is defined as
\begin{equation}
  \mathcal{P}_{\pi} = \lim_{K \to \infty} \frac{1}{K}\mathbb{E}_{\pi}\left[ \sum_{k=1}^{K}P_k \right].\label{def:PAoI}
\end{equation}
As we focus on the set of policies where the source can only submit packets when the channel is idle and the queue is empty, we define the set of all feasible scheduling policies as $\Pi = \{\pi = \{Z_i, i \ge 1\} | Z_i \ge T_i + W_i\}$,the problem of minimizing the average PAoI can be formulated as
\begin{align}
  \mathcal{P}^{*,\omega}=\min_{\pi \in \Pi}\mathcal{P}_{\pi}\label{def:prob},
\end{align}
where the superscript $\omega\equiv \text{wop}$ is the case without preemption and $\omega\equiv \text{wp}$ is the case with preemption. Our goal is to find the optimal policy in $\Pi$, denoted by $\pi^*$, that achieves the minimum average PAoI $\mathcal{P}^{*,\omega}$. In the following, we consider the optimal policy in the two case one by one.

\section{The Non-Preemption Case}\label{sec:nonpreemption}

In this section, we solve problem~\eqref{def:prob} in the case $\omega\equiv \text{wop}$. In this case, all packets is delivered successfully, hence we have $i_k=k$ for all $k=1,2,\cdots$ and $X_k=Z_{k-1}$. We replace $i_k$ with $k$ for simplicity. We solve problem~\eqref{def:prob} in the case $\omega\equiv \text{wop}$ in three steps: First, we prove that a \textit{continuously working} policy is optimal for this problem. Second, we prove that a \textit{randomized threshold} policy is optimal and third, each randomized threshold policy is outperformed by a \textit{fixed threshold} policy. Finally, we develop a low-complexity algorithm to find the optimal fixed threshold policy for special cases.

\subsection{The Optimality of Continuous Working Policies}

Based on definitions in Sec.~\ref{subsec:systemmodel}, we have $P_k=X_k+A_k=X_k+T_k+W_k+C_k$. The PAoI benefits from reducing $X_k$ and $W_k$. Intuitively, reducing the inter-generation time $X_k$ increases the waiting time $W_k$ and thus may degrade the PAoI performance. However, if an idle time exists between the delivery of the last packet $D_{k-1}$ and the generation of a new packet $S_{k}$, it is always beneficial to reduce this idle time by reducing $X_k$. We express this idea in the following lemma.
\begin{Lemma}\label{lem:maxZ}
  In the optimal policy for the case $\omega\equiv \text{wop}$, the source submits a new packet before or immediately when the old packet is delivered to the destination, which means $S_k\le D_{k-1}$ for all $k=1,2,\cdots$.
\end{Lemma}
\begin{proof}
  The details are provided in Appendix~\ref{proof:lem:maxZ}
\end{proof}

Lemma~\ref{lem:maxZ} shows that idle time in the system will increase the average PAoI. This conclusion shows the difference between PAoI and AoI as indicators of information freshness: appropriate idle time may reduce the average AoI~\cite{sun2017update} but only increases the PAoI. By substituting \eqref{equ:dk} and \eqref{equ:zk} into $S_k\le D_{k-1}$, we have the following result.
\begin{corollary}\label{cor:maxZ}
  The optimal policy for the problem~\eqref{def:prob} in the case $\omega\equiv \text{wop}$ satisfies $Z_k\le T_k+W_k+C_k$.
\end{corollary}
Based on Corollary~\ref{cor:maxZ}, we focus on the following special  policies.
\begin{definition}[Continuous Working Policy]
  A policy is said to be a Continuous Working Policy, if it satisfies $Z_k\le T_k+W_k+C_k$ for all $k$.
\end{definition}

Let $\Pi_{\text{CW}}$ ($\Pi_{\text{CW}}\subseteq  \Pi$) denote the set of continuous working policies. In the following, we focus on a special policy in $\Pi_{\text{CW}}$.

\subsection{The Optimality of Randomized Threshold Policies}\label{subsec:RT}
A policy $\pi\in \Pi_{\text{CW}}$ is said to be a \emph{randomized threshold} policy if it chooses a threshold $\Theta_k$ from an invariant distribution $f_{\Theta}(\cdot)$ and then determines $Z_k$ according to $Z_k=T_k+W_k+\min\{\Theta_k,C_k\}$. This means that the new update is generated at most $\Theta_k$ times later than the previous one starts to be computed. We use $\Pi_{\text{RT}}$ ($\Pi_{\text{RT}}\subseteq \Pi_{\text{CW}}$) to represent the set of randomized threshold policies. The first key result of this paper is the following theorem.
\begin{theorem}\label{The:WOP_RT}
  Given the distribution of transmission time and computation time, a randomized threshold policy in $\Pi_{\text{RT}}$ is optimal for Problem \eqref{def:prob} in the case $\omega\equiv \text{wop}$.
\end{theorem}
\begin{proof}
  The details are provided in Appendix~\ref{proof:The:WOP_RT}.
\end{proof}

Based on Theorem~\ref{The:WOP_RT}, we can reformulate Problem~\eqref{def:prob} as follows:
\begin{align*}
  P_{k+1}&=X_{k+1}+A_{k+1}=Z_{k}+A_{k+1}\\
  &=T_{k}+W_{k}+\min\{\Theta,C_{k}\}+T_{k+1}+W_{k+1}+C_{k+1}
\end{align*}
where $W_k=\max\{0,D_{k-1}-S_k-T_k\}=\max\{0,C_{k-1}-\Theta_{k-1}-T_k\}$. Since $T_k$'s ,$C_k$'s and $\Theta_k$'s are i.i.d., we have that $W_k$'s are also i.i.d. Thus, by dropping the subscripts from all random variables and replacing numerical average with expectation, Problem~\eqref{def:prob} can be reformulated as the following functional optimization problem:
\begin{align}
  \mathcal{P}^{*,\text{wop}}=\min_{f_{\Theta}}&\mathbb{E}\left[\min\{\Theta,C\}+2W\right]+2\mathbb{E}\left[T\right]+\mathbb{E}\left[C\right]\label{def:prob2}\\
  &s.t.\quad W=\max\{0,\widetilde{C}-\Theta-T\}.\nonumber
\end{align}
where $\widetilde{C}$ represents the computation time of the last packet. 

Theorem~\ref{The:WOP_RT} shows that in the non-preemptive system, we can optimize average PAoI without considering historical informations, such as transmission times, computation times of old packets and inter-generation times. Next, we consider the optimal policy in $\Pi_{\text{RT}}$.

\subsection{The Optimality of Fixed Threshold Policies}\label{subdec:FT}
A policy $\pi\in \Pi_{\text{RT}}$ is said to be a \emph{fixed threshold} policy, if $\Theta$ degrades to a constant value, denoted by $\theta$. In this case, $Z_k=T_k+W_k+\min\{\theta,C_k\}$ for all $k=1,2,\cdots$. Let $\Pi_{\text{FT}}$ ($\Pi_{\text{FT}}\subseteq \Pi_{\text{RT}}$) denote the set of fixed threshold policies.
\begin{theorem}\label{The:WOP_FT}
  Given the distribution of transmission time and computation time, a fixed threshold policy in $\Pi_{\text{FT}}$ is optimal for Problem~\eqref{def:prob2} in the case $\omega\equiv \text{wop}$.
\end{theorem}
\begin{proof}
  The details are provided in Appendix~\ref{proof:The:WOP_FT}.
\end{proof}
Theorem~\ref{The:WOP_FT} shows that in the non-preemptive system, the source can apply a simple policy to minimize the average PAoI: when a packet starts to be processed, the source sends the new packet after waiting for a constant time $\theta$ or after the data is processed. Thus, it is critical to determine the value of the threshold $\theta$, which is studied in the following.


Denote $\mathcal{P}^{*,\text{wop}}(\theta)$ as the average PAoI in the fixed threshold policy with parameter $\theta$. We have
\begin{align}
  \mathcal{P}^{*,\text{wop}}(\theta)=\mathbb{E}\left[\min\{\theta,C\}+2W\right]+2\mathbb{E}\left[T\right]+\mathbb{E}\left[C\right]\label{def:Ptheta},
\end{align}
where
\begin{align}
  &\mathbb{E}\left[\min\{\theta,C\}\right]=\theta\int_{\theta}^{\infty}f_C(x)  \,dx+\int_{0}^{\theta}f_C(x)x  \,dx,\label{equ:EC}\\
  &\mathbb{E}\left[W\right]=2\int_{0}^{\infty}f_T(x)\int_{x+\theta}^{\infty} f_C(y)(y-\theta-x) \,dy   \,dx.\label{equ:EW}
\end{align}
Based on Theorems \ref{The:WOP_RT} and \ref{The:WOP_FT}, the Problem~\eqref{def:prob} can be reformulated as
\begin{align}
  \mathcal{P}^{*,\text{wop}}=\min_{{\theta}\in [0,\infty)\cup \{\infty\} }\mathcal{P}^{*,\text{wop}}(\theta).\label{def:prob3}
\end{align}

Assuming  that $f_C$ is continuous\footnote{For the case of non-continuous or discrete distribution, a similar analysis can be launched by introducing the  Dirac impulse function.}, Problem~\eqref{def:prob3} is a one-dimensional optimization of a continuous, differentiable function. Thus, the minimum point is either a boundary point ($0$ or $\infty$) or a local minimum. In particular,  two special fixed threshold policies are defined: if $\theta=0$, the source sends a new packet when the old packet begins to be computed; if $\theta=\infty$, the source sends a new packet when the computing process of the old one is finished. By substituting $\theta=0$ and $\theta=\infty$ into \eqref{equ:EC} and \eqref{equ:EW}, we have
\begin{align}
  \mathcal{P}^{*,\text{wop}}(0)&=2\int_{0}^{\infty}f_T(x)\int_{x}^{\infty} f_C(y)(y-x) \,dy\,dx\nonumber\\
  &\quad +2\mathbb{E}[T]+\mathbb{E}[C]\\
  \mathcal{P}^{*,\text{wop}}(\infty)&=2\mathbb{E}[T]+2\mathbb{E}[C]
\end{align}

There may exist multiple local minima, and they all satisfy the following condition.

\begin{proposition}\label{pro:fis}
  The local minimum $\theta^{\dagger}$ of Problem~\eqref{def:prob3} satisfies
  \begin{align}
    2\mathbb{E}_T\left[F_C(x+\theta^{\dagger})\right]=F_C(\theta^{\dagger})+1,\label{equ_fi}
  \end{align}
  and
  \begin{align}
    \mathbb{E}_T\left[f_C(T+\theta^{\dagger})\right]\ge \frac{f_C(\theta^{\dagger})}{2}.\label{equ_fi2}
  \end{align}
\end{proposition}
\begin{proof}
  The equations are directly derived based on the first-order condition $\frac{d \mathcal{P}^{*,\text{wop}}(\theta)}{d \theta}=0$ and the second-order condition $\frac{d^2 \mathcal{P}^{*,\text{wop}}(\theta)}{d\theta^2}\ge 0$.
\end{proof}

For any given distributions of $T$ and $C$, the local minimum $\theta^{\dagger}$ can be found by solving \eqref{equ_fi} and checking if \eqref{equ_fi2} is satisfied. The optimal threshold $\theta^*$ is the value among $0,\theta^{\dagger},\infty$ that achieves the minimum average PAoI $\mathcal{P}^{*,\text{wop}}=\min\{\mathcal{P}^{*,\text{wop}}(0),\mathcal{P}^{*,\text{wop}}(\theta^{\dagger}),\mathcal{P}^{*,\text{wop}}(\infty)\}$.

\subsection{Special Cases}\label{subsec:specialcase}

In the following, we focus on two special cases: 1) when the computation time follows an exponential distribution, the optimal policy $\theta^*$ is always either $0$ or $\infty$; 2) when the transmission time follows an exponential distribution, we can calculate $\theta^{*}$ using a low complexity algorithm developed based on equations \eqref{equ_fi} and \eqref{equ_fi2}.


\subsubsection{Exponentially Distributed Computation Time}

If the computation time follows an exponential distribution with parameter $\mu$, which means $f_C(x)=\mu e^{-\mu x}, x > 0$. From \eqref{equ:EC} and \eqref{equ:EW}, we have $\mathcal{P}^{*,\text{wop}}(\theta)$ is first-order differentiable. Denote $P'(\theta)$ as the first derivative of $\mathcal{P}^{*,\text{wop}}(\theta)$, then we have
\begin{align*}
  P'(\theta)&=\int_{\theta}^{\infty}f_C(x)  \,dx-2\int_{0}^{\infty}f_T(x)\int_{x+\theta}^{\infty} f_C(y) \,dy   \,dx\\
  &=e^{-\mu \theta}(1-2\int_{0}^{\infty}f_T(x)e^{-\mu x}\,dx)\\
  &=e^{-\mu \theta}(1-2 \mathcal{L}_{\mu}),
\end{align*}
where $\mathcal{L}_{\mu}$ is the Laplace transform of the transmission time distribution. It can be easily seen that $P'(\theta) \neq 0$ for any $\theta$ if $\mathcal{L}_{\mu} \neq 1/2$. Therefore, the minimum must be achieved at the boundaries. In particular, $\theta^*=0$ if $\mathcal{L}_{\mu}\le \frac{1}{2}$ and $\theta^*=\infty$ if $\mathcal{L}_{\mu}> \frac{1}{2}$.

\subsubsection{Exponentially Distributed Transmission Time}

Assume that the transmission time follows an exponential distribution with parameter $\lambda$, which means $f_T(x)=\lambda e^{-\lambda x}, x > 0$. We denote the second derivative of $\mathcal{P}^{*,\text{wop}}$ by $P''(\theta)$. In general, there is no closed-form expression for  $\theta^{*}$ for a continuous computation time distribution. However, we can first numerically calculate $\theta^{\dagger}$ and then compare $\mathcal{P}^{*,\text{wop}}(\theta^{\dagger}),\mathcal{P}^{*,\text{wop}}(0)$, and $\mathcal{P}^{*,\text{wop}}(\infty)$ to obtain $\theta^*$. 

To calculate $\theta^{\dagger}$, we need to find all $\theta\in (0,\infty)$ that satisfy \eqref{equ_fi} and \eqref{equ_fi2}. To achieve this, we define 
\begin{align}
  P''_{z}(\theta)=\lambda(1-F_C(\theta))-f_C(\theta).\label{equ:pz}
\end{align}
Then, we have the following lemma.
\begin{Lemma}\label{Lem:expPz}
  When the transmission time follows an exponential distribution, for any $\theta$ that satisfies \eqref{equ_fi}, we have
    \begin{align}
      P''(\theta)=P''_{z}(\theta).
    \end{align}
\end{Lemma}
\begin{proof} 
  The details are provided in Appendix~\ref{proof:Lem:expPz}.
\end{proof}

Since $P_z''(\theta) = 0$ is generally easier to solve compared with~\eqref{equ_fi}, we can calculate the local minimum based on $P''_z(\theta)$ as shown in the following proposition.

\begin{proposition}\label{Lem:pz}
  When the transmission time follows an exponential distribution, suppose $0\le \theta_1 <\theta_2<\infty$, we have
  \begin{itemize}
    \item If $P''_z(\theta)>0$ for all $\theta_1<\theta<\theta_2$, there is at most one $\theta\in (\theta_1,\theta_2)$ that satisfies \eqref{equ_fi}, and if such a $\theta$ exists, it also satisfies \eqref{equ_fi2}.
    \item If $P''_z  (\theta)<0$ for all $\theta_1<\theta<\theta_2$, there is no $\theta\in [\theta_1,\theta_2]$ that satisfies \eqref{equ_fi}.
  \end{itemize} 
\end{proposition}
\begin{proof}
  The details are provided in Appendix~\ref{proof:lempz}.
\end{proof}

Based on Proposition~\ref{Lem:pz}, we can calculate the local minimum as follows. Denote $\theta_1,\theta_2,\cdots$ as all zero points of $P''_z(\theta)$. These zero points divide $(0,\infty)$ into multiple sub-intervals. On each sub-interval, if $P''_z(\theta)>0$, we can use the bisection method to calculate $\theta$ which satisfies \eqref{equ_fi} and \eqref{equ_fi2}. Then, by comparing the corresponding average PAoI of these $\theta$'s, we can obtain $\theta^{\dagger}$, and finally obtain $\theta^*$. The procedure is shown in Algorithm~\ref{alg:optimal}.

\begin{algorithm}[!t]
  \caption{Piece-wise Bisection Method for Solving Problem~\eqref{def:prob3} for Exponentially Distributed Transmission Time}\label{alg:optimal}
  \begin{algorithmic}[1]
      \STATE $\textbf{given }\theta_1,\cdots,\theta_M$
      \STATE $\mathbb{S}=\emptyset$
      \STATE $\textbf{for }(l,r)=\{(0,\theta_1),\cdots,(\theta_{M},\infty)\}$
      \STATE \hspace{0.5cm}$\textbf{if }P''_z(\theta)>0,\theta\in (l,r)$
      \STATE \hspace{1cm} find $\theta$ by bisection method such that \eqref{equ_fi} holds
      \STATE \hspace{1cm} $\mathbb{S}=\mathbb{S}\cup \{\mathcal{P}^{*,\text{wop}}(\theta)\}$
      \STATE $\theta^{\dagger}={\arg\min}_{\theta}\mathbb{S}$
      \STATE $\theta^{*}={\arg\min}_{\theta}\{\mathcal{P}^{*,\text{wop}}(\theta^{\dagger}),\mathcal{P}^{*,\text{wop}}(0),\mathcal{P}^{*,\text{wop}}(\infty)\}$
      \STATE $\textbf{return }\theta^*$
  \end{algorithmic}
\end{algorithm}

\section{The Service Preemption Case}\label{sec:preemption}

In this section, we focus on Problem~\eqref{def:prob} in the case $\omega\equiv \text{wp}$. In this case, we have $W_i=0$ for all $i=1,2,\cdots$. We begin by proving that the continuous working policy is optimal in this scenario. Afterwards, we prove that the optimal policy is a \textit{stationary deterministic} policy. Lastly, we discuss two heuristic policies.

\subsection{The Optimality of Continuous Working Policies}

In the preemptive case, we have
\begin{align}
  P_k=&X_k+A_k=S_{i_k}-S_{i_{k-1}}+T_{i_k}+C_{i_k}\nonumber\\
  =&\sum_{j=i_{k-1}}^{i_k-1}(S_{j+1}-S_{j})+T_{i_k}+C_{i_k}\nonumber\\
  =&\sum_{j=i_{k-1}}^{i_k-1}Z_j+T_{i_k}+C_{i_k}.\label{equ:WP_Pk1}
\end{align}
Unlike the case there $\omega\equiv \text{wop}$, reducing the inter-generation time $Z_i$ does not increase $W_i$, but may increase preemption and then degrade the PAoI performance. However, reducing idle time between packet $i_k$'s delivery time $D_{i_k}$ and packet $(i_k+1)$'s submission time $S_{i_k+1}$ is still beneficial. The following corollary shows that the continuous working policy is also optimal in this case.
\begin{corollary}\label{cor:maxZ_wp}
  The optimal policy $\pi^*$ for Problem \eqref{def:prob} in the case $\omega\equiv \text{wp}$ is the continuous working policy, which satisfies $Z_i\le T_i+C_i$.
\end{corollary}
\begin{proof}
  The proof follows the same lines as the proof of Lemma~\ref{lem:maxZ}, except that $k$ is replaced by $i_k$ and $W_{i}=0$ for $i=1,2,\cdots$.
\end{proof}

Corollary~\ref{cor:maxZ_wp} demonstrates that the similarities between the preemptive system and the non-preemptive system: idle time will damage the average PAoI performance. In the following, we consider the optimal policy in $\Pi_{\text{CW}}$.

\subsection{The Optimality of Stationary Deterministic Policies}

A policy $\pi\in \Pi_{\text{CW}}$ is said to be a \emph{stationary deterministic} policy, if it chooses a threshold $\theta_i$ based on $T_i$, and then determines $Z_i$ according to $Z_i=T_i+\min\{\theta_i,C_i\}$. In this policy, $\theta_i$ is determined according to a function $g(T_i)$ that is invariant for all $i$. Let $\Pi_{\text{SD}}$ ($\Pi_{SD}\subseteq \Pi_{\text{CW}}$) denote the set of stationary deterministic policies. The following theorem shows that the optimal solution for Problem \eqref{def:prob} in this case is a stationary deterministic policy.

\begin{theorem}\label{The:WP_SD}
  Given the distribution of transmission time and computation time, a stationary deterministic policy in $\Pi_{\text{SD}}$ is optimal for Problem \eqref{def:prob} in the case $\omega\equiv \text{wp}$.
\end{theorem}
\begin{proof}
  The details are provided in Appendix~\ref{proof:The:WP_SD}.
\end{proof}

According to Theorem \ref{The:WP_SD}, Problem \eqref{def:prob} can be reformulated in $\Pi_{\text{SD}}$ as follows.

Using \eqref{equ:WP_Pk1}, we obtain
\begin{align}
  P_k=&\sum_{j=i_{k-1}}^{i_k-1}Z_j+T_{i_k}+C_{i_k}\nonumber\\
  =&\sum_{j=i_{k-1}}^{i_k-1}(Z_j+(T_{j+1}+C_{j+1})\mathbf{1}_{\Omega_{j+1}}),\label{equ:Pk_}
\end{align}
where $\Omega_j:C_j\le \theta_j+T_{j+1}$ is the event that packet $j$ is successfully delivered, and $\mathbf{1}_{\Omega_j}$ is the indicator function of event $\Omega_j$. Under the stationary deterministic policy, $P_k$'s are i.i.d. Denote $M_k$ as the number of generated packets between two  successful packets, which follows the geometric distribution with $p=Pr(\Omega_i)$. Similar to the previous section, by dropping the subscripts of all variables, we have
\begin{align}
  \mathbb{E}[P]=&\mathbb{E}\left[\sum_{j=1}^{M}(Z_j+(T_{j+1}+C_{j+1})\mathbf{1}_{\Omega_{j+1}})\right]\label{_equ:Pa}\\
  =&\mathbb{E}\left[M\right]\mathbb{E}\left[Z+(T'+C')\mathbf{1}_{\Omega'}\right]\label{_equ:Pb}\\
  =&\mathbb{E}\left[M\right]\mathbb{E}\left[Z+(T+C)\mathbf{1}_{\Omega}\right]\label{_equ:Pc}\\
  =&\frac{\mathbb{E}\left[Z+(T+C)\mathbf{1}_{\Omega}\right]}{Pr(\Omega)},\label{_equ:Pd}
\end{align}
where \eqref{_equ:Pb} follows from Wald's Equation (\cite{ross1995stoch}, Chapter 6); \eqref{_equ:Pc} follows that $T_i$'s and $C_i$'s are i.i.d.; and \eqref{_equ:Pd} follows that the expectation of geometrically distributed random variable $M$ is $1/p$.

Problem~\eqref{def:prob} in this case can be reformulated as a  functional optimization problem as follows:
\begin{align}
  \mathcal{P}^{*,\text{wp}}=\min_{g:\theta=g(T)}&\frac{\mathbb{E}\left[T+\min\{\theta,C\}+(T+C)\mathbf{1}_{\Omega}\right]}{Pr(\Omega)}\label{def:probX},
\end{align}
where
\begin{align}
  &Pr(\Omega)=\int_{0}^{\infty} f_{T}(y)\int_{0}^{\infty} f_{T}(x)F_{C}(y+g(x)) \,dx \,dy,\label{equ:PO}\\
  &\mathbb{E}\left[(T+C)\mathbf{1}_{\Omega}\right]=\nonumber\\
  &\int_{0}^{\infty} f_{T}(s)\int_{0}^{\infty}f_T(x)\int_{0}^{s+g(x)} f_{C}(y)(x+y)\,dy \,dx  \,ds.\label{equ:ETC}
\end{align}


Previous studies on the functional optimization problem have mainly focused on the convex case, where the optimal policy generally has a threshold structure. However, Problem \eqref{def:probX} is non-convex and challenging to solve. Additionally, there is no evidence to suggest that the optimal solution has a specific structure, such as a threshold policy. In this paper, we present two heuristic solutions to tackle the problem.

\subsection{The Fixed Threshold Policy}\label{subsec:fixed}

\begin{algorithm}[!t]
  \caption{Iterative Algorithm for Obtaining $\theta^*_f$}\label{alg:itera}
  \begin{algorithmic}[1]
      \STATE $\textbf{given a sufficiently small tolerance } \delta \textbf{ and arbitrary } c$;
      \WHILE {not convergence}
      \STATE Solve the problem~\eqref{def:probX_const_lag} for a given $c$ and obtain the corresponding optimal threshold $\theta^*_c$;
      \IF {$-\delta\le p(c)\le \delta$}
      \STATE Convergence=true;
      \STATE $\textbf{return }\theta^*_f=\theta^*_c$
      \ELSE
      \STATE Update $c=\mathcal{P}_f^{*,\text{wp}}(\theta^*_c)$
      \ENDIF
      \ENDWHILE
  \end{algorithmic}
\end{algorithm}

Inspired by the optimal fixed threshold policy for the non-preemptive case, we consider the fixed threshold policy for the preemptive case. This makes sense because the fixed threshold policy is a special stationary deterministic policy. Under this policy, a new packet is submitted a fixed time after the transmission completion of the last packet, i.e., $Z=T+\min\{\theta,C\}$ and $\theta$ is a constant value. Let $\mathcal{P}^{*,\text{wp}}_{\text{f}}(\theta)$ denote the average PAoI and $\theta^*_f$ denote the optimal threshold in this case. Then, Problem~\eqref{def:probX} can be reformulated as follows:
\begin{align}
  \mathcal{P}^{*,\text{wp}}_{\text{f}}=&\min_{{\theta}\in [0,\infty)\cup \{\infty\} }\mathcal{P}^{*,\text{wp}}_{\text{f}}(\theta)\label{def:probX_const}
\end{align}

To make problem~\eqref{def:probX_const} more easier to solve, we introduce the following parameterized problem:
\begin{align}
  p(c)\triangleq \min_{{\theta}\in [0,\infty)\cup \{\infty\} } &\mathbb{E}\left[T+\min\{\theta,C\}+(T+C)\mathbf{1}_{\Omega}\right]\nonumber\\
  &-cPr(\Omega)\label{def:probX_const_lag}.
\end{align}

This approach was previously discussed in \cite{ref2}. We now present a restatement of the results in \cite{ref2} as the following lemma, and provide a proof for completeness.
\begin{Lemma}\label{Lem:pc}
  $p(c)$ is decreasing in $c$, and the optimal solution of problem~\eqref{def:probX_const} is given by $c^*$ that solves $p(c^*)=0$.
\end{Lemma}
\begin{proof}
  The details are provided in Appendix~\ref{proof:Lem:pc}
\end{proof}

Lemma~\ref{Lem:pc} demonstrates that for the considered optimization problem with a fractional form objective function, an equivalent optimization problem exists with a subtraction form objective function, i.e., $p(c)$, and both forms of the problem yield the same scheduling policy. An iterative algorithm with guaranteed convergence in \cite{ref2} can be employed to obtain $c^*$, which can be found in Algorithm~\ref{alg:itera}.

During the iteration, in order to obtain $\theta^*_f$, we need to
solve problem~\eqref{def:probX_const_lag}. For a certain value of $c$, the optimal threshold of problem~\eqref{def:probX_const_lag}, denoted by $\theta^*_c$, can be found similarly to Sec.~\ref{subdec:FT}. Specifically, based on \eqref{equ:EC},\eqref{equ:PO} and \eqref{equ:ETC}, the objective function of \eqref{def:probX_const_lag} is second-order differentiable when $f_c(x)$ is first-order differentiable, and then the local minimum $\theta_{c}^{\dagger}$ satisfies the first-order and second-order conditions shown in \eqref{equ:fi3} and \eqref{equ:fi4}. Once obtaining $\theta_{c}^{\dagger}$, we can calculate  $\theta^*_c$ as the one among $0,\theta^{\dagger}_c,\infty$ that achieves the minimum $\mathbb{E}\left[T+\min\{\theta,C\}+(T+C)\mathbf{1}_{\Omega}\right]-cPr(\Omega)$.

\begin{figure*}[!t]
  \normalsize
  \begin{align}
    &\int_{0}^{\infty} f_{T}(y)\int_{0}^{\infty}f_T(x)f_C(y+\theta)(x+y+\theta-c) \,dx  \,dy+\int_{0}^{\infty}f_T(x)\int_{\theta}^{\infty}f_C(y)\,dy \,dx=0\label{equ:fi3} \\
    &\int_{0}^{\infty} f_{T}(y)\int_{0}^{\infty}f_T(x)\left[f_C(y+\theta)+f_C'(y+\theta)(x+y+\theta-c)\right] \,dx \,dy\ge \int_{0}^{\infty}f_T(x)f_C(\theta) \,dx\label{equ:fi4}
  \end{align}
  \hrulefill
  \vspace*{2pt}
\end{figure*}

\subsection{Transmission-Aware Threshold Policy}

We now investigate a transmission-aware threshold policy, where the waiting time for a new packet submission  is inversely proportional to the transmission time of the last packet. Specifically, we set $g(T)=\max\{0,\beta-T\}$, where $\beta$ is a threshold value. Then, we have $Z=\min\{T+C,\max\{\beta,T\}\}$. Based on the definition of $\Pi_{\text{SD}}$, we know that this policy is also a special stationary deterministic policy.

Let $\mathcal{P}^{*,\text{wp}}_{\text{t}}(\beta)$ denote the average PAoI under the transmission-aware threshold policy and $\beta^*_t$ denote the optimal solution in this case. We can reformulate Problem~\eqref{def:probX} as a one-dimensional fractional optimization problem as follows:
\begin{align}
  &\min_{{\beta}\in [0,\infty)\cup \{\infty\} }\mathcal{P}^{*,\text{wp}}_{\text{t}}(\beta)\label{def:probX_thres}.
\end{align}
We can solve the above problem in a manner similar to that in Sec. IV-C.

In the preemptive system, the transmission-aware threshold policy generally outperforms the fixed threshold policy, unlike in  the non-preemptive system. In particular, when the computation time follows an exponential distribution, the transmission-aware threshold policy is the optimal policy for Problem~\eqref{def:probX}, as stated in the following theorem:


\begin{theorem}\label{The:optiamlC}
  When the computation time is exponentially distributed, the optimal solution for Problem~\eqref{def:probX} is a transmission-aware threshold policy, which is given by
  \begin{align}
    g(x)=\max\{0,\gamma -x\}
  \end{align}
  with threshold given by
  \begin{align}
    \gamma =\frac{\mathcal{L}_{\mu}+\mu \mathcal{P}^{*,\text{wp}}\mathcal{L}_{\mu}-\mu\mathcal{M}_{\mu}-2}{\mu\mathcal{L}_{\mu}}
  \end{align}
  where $\mathcal{L}_{\mu}=\int_{0}^{\infty}f_T(x)e^{-\mu x}\,dx,\mathcal{M}_{\mu}=\int_{0}^{\infty}xf_T(x)e^{-\mu x}\,dx$ and $\mu$ is the parameter of exponential distribution.
\end{theorem}
\begin{proof}
  The details are provided in Appendix~\ref{proof:optimalC}.
\end{proof}

It is noteworthy that the optimality of the transmission-aware policy is generally difficult to prove. However, numerical results in the next section will show that the transmission-aware policy usually performs better than the fixed threshold policy in the preemptive system.

\section{Numerical Results}\label{sec:num}

In this section, we analyze the performance of optimal policies for both systems under the assumption that the computation time follows an exponential distribution, and the transmission time follows exponential and Pareto distributions. We set $\mathbb{E}[T+C]=1$ and vary the ratio between $\mathbb{E}[T]$ and $\mathbb{E}[C]$ in our simulations. To establish a baseline for comparison, we adopt the mean-threshold policy \cite{cham2021min}, which employs the fixed threshold policy and uses the mean of the computation time as the threshold. Additionally, we observed the performance of the policies proposed in this paper in terms of AoI and compared them with other reference policies.


\begin{figure}[!t]
  \centering
  \includegraphics[width=0.9\linewidth]{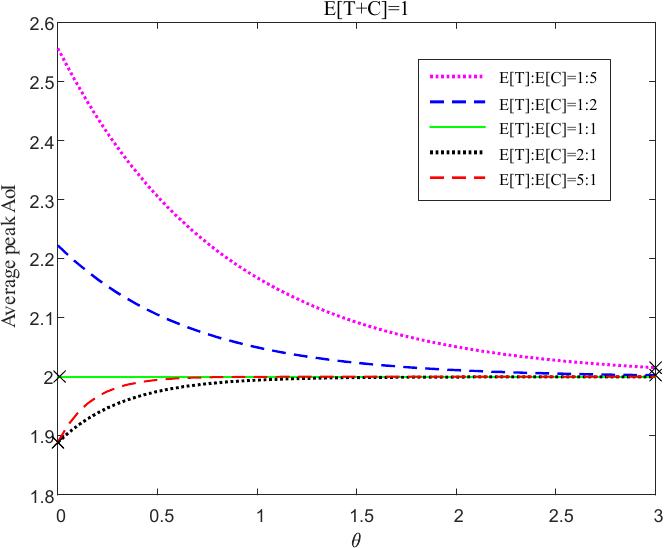}
  \caption{Average PAoI vs. $\theta$ under the fixed threshold policy with exponential transmission and computation time for different $\mathbb{E}[T]:\mathbb{E}[C]$ in the non-preemptive system.} \label{fig:exp16}
\end{figure}

\begin{figure}[!t]
  \centering
  \includegraphics[width=0.9\linewidth]{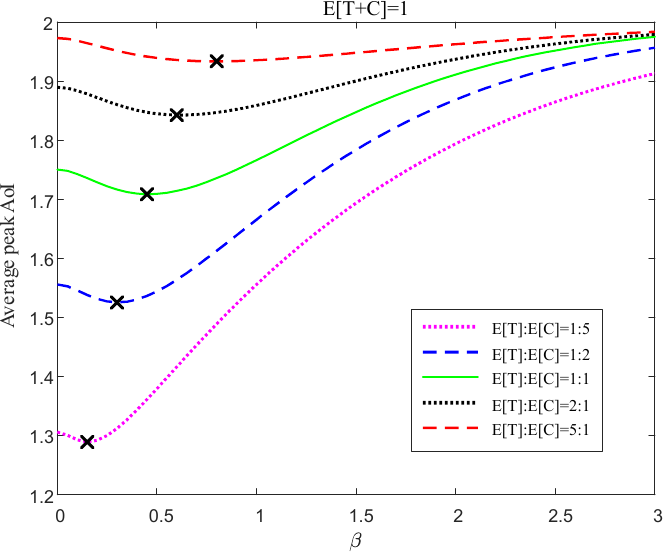}
  \caption{Average PAoI vs. $\beta$ under the transmission-aware threshold policy with exponential transmission and computation time for different $\mathbb{E}[T]:\mathbb{E}[C]$ in the service preemptive system.} \label{fig:exp17}
\end{figure}


\begin{figure}[!t]
  \centering
  \includegraphics[width=0.9\linewidth]{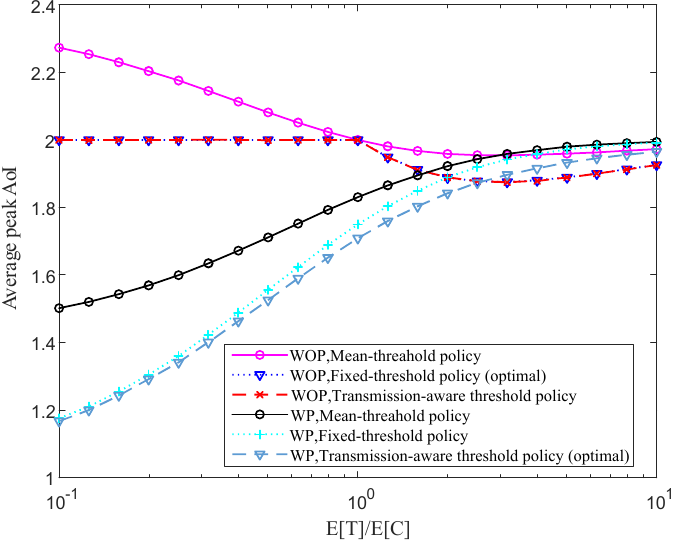}
  \caption{Average PAoI achieved by different policies  under exponential transmission and computation time with varying $\frac{\mathbb{E}[T]}{\mathbb{E}[C]}$ in two systems.}\label{fig:exp18}
\end{figure}

\subsection{Exponential Transmission Time and Computing Time}

In this example, we consider both transmission and computation times to be exponentially distributed with parameters $\lambda$ and $\mu$, respectively. The mean transmission time and computation time are given by $\mathbb{E}[T]=\frac{1}{\lambda}$ and $\mathbb{E}[C]=\frac{1}{\mu}$, respectively. 

In Fig.~\ref{fig:exp16}, we present the average PAoI under the fixed threshold policy for different  $\frac{\mathbb{E}[T]}{\mathbb{E}[C]}$ values. The optimal threshold $\theta^*$ that achieves the minimum PAoI is indicated by black cross signs. As discussed in Sec.\ref{subsec:specialcase}, when $ \frac{\mathbb{E}[T]}{\mathbb{E}[C]}>1$, we have $\lambda:\mu<1$, and then the Laplace transform $\mathcal{L}_{\mu}=\frac{\lambda}{\lambda+\mu}<\frac{1}{2}$. Consequently, the average PAoI performance of the optimal policy strictly decreases with increasing $\theta$, and the optimal threshold is $\theta^*=\infty$. Conversely, when $ \frac{\mathbb{E}[T]}{\mathbb{E}[C]}<1$, as shown in Fig.~\ref{fig:exp16}, the average PAoI is increasing and $\theta^*=0$. When $ \frac{\mathbb{E}[T]}{\mathbb{E}[C]}=1$, the average PAoI is constant for all threshold values. Theoretical and experimental results both indicate that the optimal threshold $\theta^*$ decreases as $ \frac{\mathbb{E}[T]}{\mathbb{E}[C]}$ increases.


Fig.~\ref{fig:exp17} illustrates the average PAoI achieved by the transmission-aware threshold policy in the preemptive system, where the optimal threshold $\beta^*$ is marked by black cross signs. As $\mathbb{E}[T]:\mathbb{E}[C]$ increases, both the optimal threshold $\beta^*$ and the minimum average PAoI increase. The pattern of the optimal threshold is completely opposite to that observed in the non-preemptive system. This indicates a fundamental difference between preemptive and non-preemptive systems. In the preemptive system, the impact of transmission delay on PAoI is greater than that of computation time. Even if the computation time is reduced, the PAoI performance of the system still decreases when the transmission time becomes longer. Moreover, new updates can only be generated after the old updates have been received by the server, resulting in a slower generation of new updates. While, in the non-preemptive system, as the transmission time increases and the computation time decreases, the source should decrease the threshold and submit new data more frequently.



In Fig.~\ref{fig:exp18}, we compare the performance of three policies: fixed threshold policy, transmission-aware threshold policy, and median threshold policy. The results reveal opposite trends in the two systems. In the preemptive system, the minimum PAoI strictly increases with an increase in $\frac{\mathbb{E}[T]}{\mathbb{E}[C]}$, whereas in the non-preemptive system, the minimum PAoI first decreases and then increases. These findings suggest that preemptive servers can better utilize dense data than non-preemptive systems. As the transmission time decreases, more frequent data submissions occur due to the decrease in the optimal threshold. Even if the computation time increases, the PAoI performance of preemptive systems still improves. In contrast, in a non-preemptive system, the minimum PAoI exhibits a U-shaped trend with respect to $\frac{\mathbb{E}[T]}{\mathbb{E}[C]}$. When the transmission time is too short, frequent data submissions increase the waiting time in the queue, resulting in a decrease in PAoI performance. Similarly, when the transmission time is too long, slow data submissions also lower PAoI performance. Only an appropriate ratio of transmission time to computation time can achieve the minimum PAoI of the system.

We also observed a counterintuitive phenomenon that the performance of non-preemptive systems can even surpass that of preemptive systems when $\frac{\mathbb{E}[T]}{\mathbb{E}[C]}$ is relatively large. In general, discarding old data to serve new data is beneficial when computation time follows an exponential distribution. However, experimental results show that retaining old data is a better choice for PAoI when transmission time is large and data submissions are infrequent.

\subsection{Pareto Transmission Time and Exponential Computing Time}

\begin{figure}[!t]
  \centering
  \includegraphics[width=0.9\linewidth]{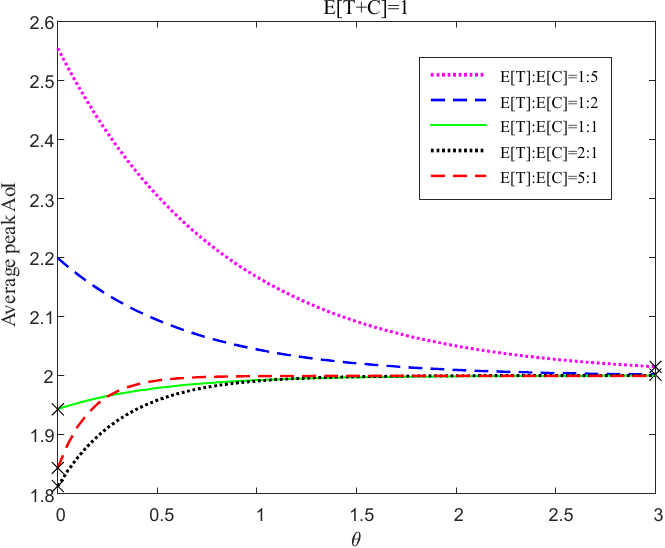}
  \caption{Average PAoI vs. $\theta$ under Pareto transmission and exponential computation time for different $E[T]:E[C]$ in the non-preemptive system.} \label{fig:exp19}
\end{figure}

\begin{figure}[!t]
  \centering
  \includegraphics[width=0.9\linewidth]{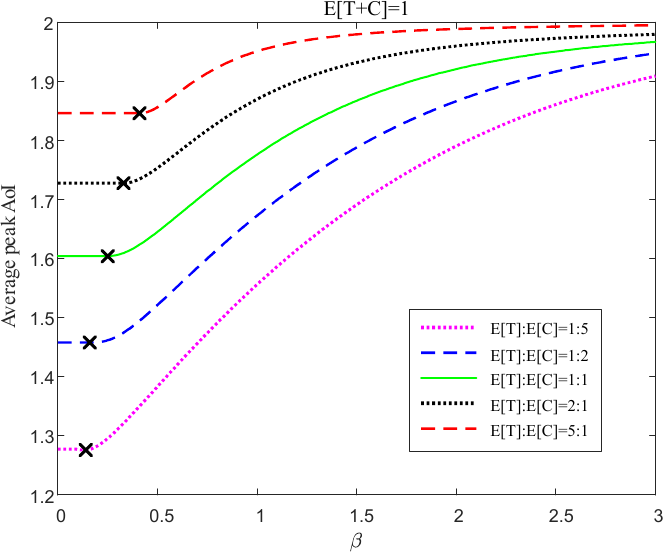}
  \caption{Average PAoI vs. $\beta$ under Pareto transmission and exponential computation time for different $E[T]:E[C]$ in the service preemptive system.} \label{fig:exp20}
\end{figure}

\begin{figure}[!t]
  \centering
  \includegraphics[width=0.9\linewidth]{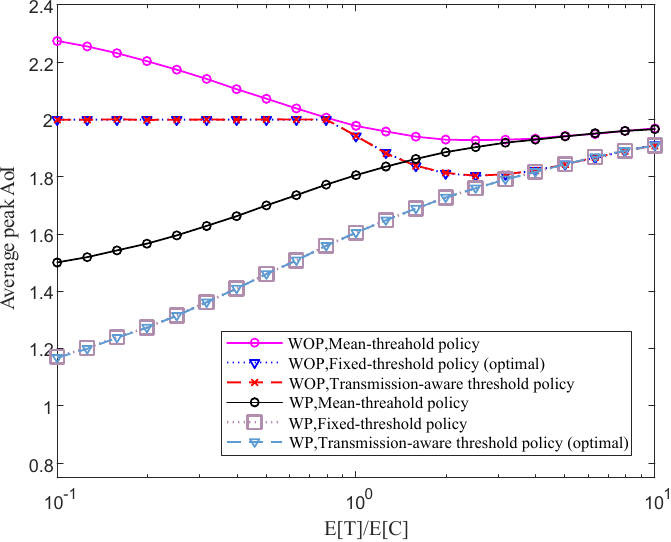}
  \caption{Average PAoI achieved by different policies  under Pareto transmission and exponential computation time with varying $\frac{\mathbb{E}[T]}{\mathbb{E}[C]}$ in two systems.} \label{fig:exp21}
\end{figure}

\begin{figure}[!t]
  \centering
  \includegraphics[width=0.9\linewidth]{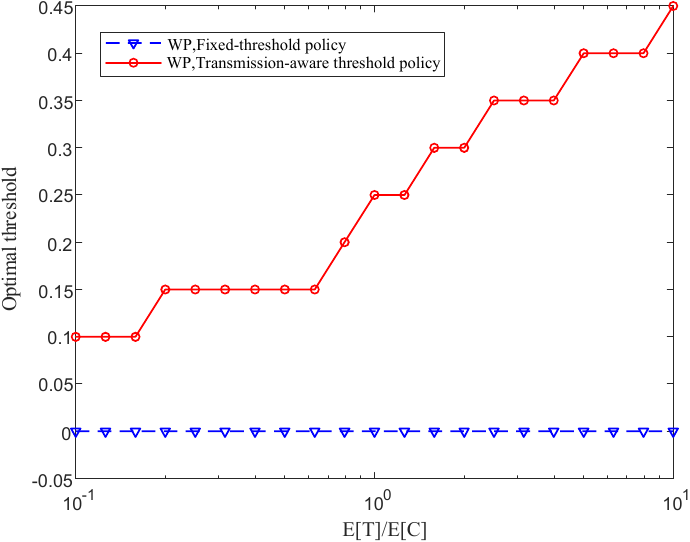}
  \caption{Optimal thresholds in different policies under Pareto transmission and exponential computation time with varying $\frac{\mathbb{E}[T]}{\mathbb{E}[C]}$ in preemptive system.} \label{fig:exp212}
\end{figure}

In this example, we assume that the transmission time follows a Pareto distribution, which is characterized by the parameters $(x_m,\alpha)$. The parameter $x_m$ represents the scale parameter, and $\alpha$ represents the shape parameter. The Pareto distribution is a heavy-tailed distribution. The smaller the $\alpha$, the heavier the tail. We choose this distribution based on the fact that the transmission time of updates is often related to the update length, which often follows a Pareto distribution. The probability density function of the Pareto distribution is given by $f(x)=\frac{\alpha x_m^{\alpha}}{x^{\alpha+1}},x>x_m$. We set $\alpha=2$ and vary $\mathbb{E}[C]$ by changing $x_m$, where $\mathbb{E}[C]$=$\frac{\alpha x_m}{\alpha-1}$. Fig.~\ref{fig:exp19} shows the average PAoI in the non-preemptive system under the optimal policy for different values of $\mathbb{E}[T]:\mathbb{E}[C]$. In this figure, the optimal threshold $\theta^*$ takes on the value of 0 when $\mathbb{E}[T]\le \mathbb{E}[C]$ and $\infty$ when $\mathbb{E}[T]>\mathbb{E}[C]$. 

Fig.~\ref{fig:exp20} displays the minimum average PAoI in the preemptive system. Due to the fact that the minimum transmission time is $x_m$, when the threshold $\beta$ is less than $x_m$, the transmission-aware policy reduces to a best-effort transmission policy, where the source sends a new update as soon as the old update arrives at the queue. Therefore, each curve in Fig.~\ref{fig:exp20} has a parallel segment on the far left. Once again, we observe that the optimal threshold $\beta^*$ and the minimum average PAoI both increase as $\frac{\mathbb{E}[T]}{\mathbb{E}[C]}$ increases. Moreover, we observe that $\beta^*$ is located at the far right end of the parallel segment, where the minimum average PAoI is close to that achieved by the best-effort policy. This finding is confirmed in Fig.~\ref{fig:exp21}.


Fig.~\ref{fig:exp21} shows the performance of the fixed threshold policy, the transmission-aware threshold policy, and the median threshold policy in two systems, considering a Pareto distribution transmission time. Similar to Fig.~\ref{fig:exp18}, the minimum PAoI follows a similar trend: it increases with $\frac{\mathbb{E}[T]}{\mathbb{E}[C]}$ in the preemptive system and initially decreases then increases in the non-preemptive system. However, different results from those in Fig.~\ref{fig:exp18} are observed in this case: 1) The minimum PAoI in the preemptive system is consistently lower than in the non-preemptive system; 2) In the preemptive system, the optimal policy and the optimal fixed threshold policy perform similarly for all $\frac{\mathbb{E}[T]}{\mathbb{E}[C]}$ values. The first result, combined with our observation in Fig.~\ref{fig:exp18}, suggests that the preemptive system are not always superior to non-preemptive systems in terms of PAoI. The second result confirms our observation in Fig.~\ref{fig:exp20}: the performance of the optimal policy is close to that of the best-effort policy. Fig.~\ref{fig:exp212} gives the optimal thresholds in fix threshold policy and transmission-aware threshold policy. We can see that the optimal fixed threshold policy has an optimal threshold $\theta^*$ of 0, at which point it also degenerates to the best-effort policy. This explains why it has similar performance to the optimal policy.

\subsection{Average AoI Performance}


In this example, we consider the AoI performance of fixed threshold and transmission-aware policies and compare them with the following two reference policies:
\begin{enumerate}
  \item Single-process policy: This policy treats transmission and computation as a single process, sending new data only after the old data has been delivered. The optimal policy under this setup is given in \cite{sun2017update}.
  \item Peak Age Threshold Policy With Postponed Plan (PAoI-TP) policy: This is a heuristic policy proposed in our previous work~\cite{zhu2022online} for the non-preemptive system. Specifically, the source estimates the peak age of new data in real-time and sends it only when the estimated PAoI exceeds a threshold and satisfies a Postponed Plan condition.
\end{enumerate}
Similarly, we consider both transmission and computation times to be exponentially distributed with parameters $\lambda$ and $\mu$, respectively. Fig.~\ref{fig:expaoi} illustrates the AoI performance of different policies in two systems. It is noteworthy that, due to the single-process considering transmission and computation as a unified process, each data is successfully delivered without entering the queue, resulting in identical performance in both systems. This explains why it corresponds to only one curve.
  
Observations reveal that in the non-preemptive system, the transmission-aware policy consistently outperforms the fixed threshold policy, contrary to experiments using PAoI as the metric. When $\mathbb{E}[T]/\mathbb{E}[C]<1$, the single-process policy outperforms the transmission-aware policy, and vice versa when $\mathbb{E}[T]/\mathbb{E}[C]>1$. This indicates that the advantages of joint transmission-computation optimization are more pronounced when the mean transmission time is greater than the mean computation time. Additionally, the AoI performance of the PAoI-TP policy always surpasses other policies. This suggests that while policies optimizing for PAoI as a goal may not simultaneously minimize AoI, we can directly use PAoI-designed policies to optimize AoI. On the other hand, in the preemptive system, the transmission-aware policy consistently outperforms both the fixed threshold policy and the single-process policy. This aligns with experiments using PAoI as the metric, highlighting the advantages of joint transmission-computation optimization in the preemptive system.

\begin{figure}[!t]
  \centering
  \includegraphics[width=0.9\linewidth]{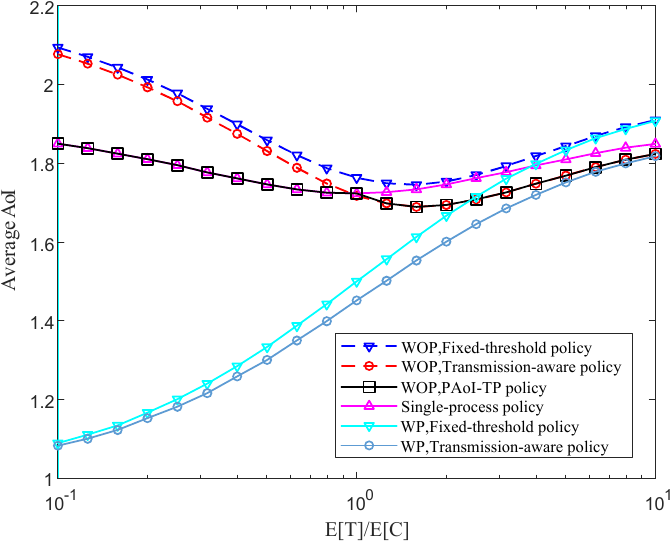}
  \caption{Average AoI achieved by different policies under exponential transmission and computation time with varying $\frac{\mathbb{E}[T]}{\mathbb{E}[C]}$ in two systems.} \label{fig:expaoi}
\end{figure}

\section{Conclusions}\label{sec:con}


In this paper, we address the problem of minimizing the PAoI in MEC systems, accounting for both transmission delays and computation times, in both preemptive and non-preemptive scenarios. We show that in the preemptive system, the fixed threshold policy achieves the minimum PAoI, while in the non-preemptive system, the transmission-aware threshold policy can achieve this minimum under the special condition of exponentially distributed computation time. The optimal threshold can be obtained either in closed-form or via a low-complexity searching algorithm, depending on the distributions of the transmission time and computation time. Our experimental results validate our theoretical analysis, and shows that the optimal threshold increases in preemptive systems but decreases in non-preemptive systems as the ratio of mean transmission time to mean computation time increases. Our analytical findings can offer valuable insights for designing scheduling algorithms in MEC systems. Future work can focus on optimizing AoI, multi-source systems, and multi-user MEC systems.

\appendix

\subsection{Proof of Lemma \ref{lem:maxZ}}\label{proof:lem:maxZ}
We prove this lemma by contradiction. Suppose that under the optimal policy $\pi\in \Pi$, packet $k+1$ is submitted $G_{k}$ seconds after  packet $k$ is delivered, which means $S_{k+1}=D_{k}+G_{k}$. using $X_k=S_k-S_{k-1}$, we can obtain the following expression:
\begin{align*}
  P_{k}&=X_k+T_k+W_k+C_k\\
  &=S_k-S_{k-1}+T_k+W_k+C_k.
\end{align*}
Now we construct a new policy $\pi'$ that  satisfies $S^{'}_j=S_j$ for all $j<k$ and $S^{'}_j=S_j-G_{k}$ for all $j \ge k$. In this policy, $G^{'}_k=0$. The new policy has no impact on packets before packet $k$, thus $P^{'}_j=P_j$ for all $j<k$. The new policy $\pi'$ submits packet $k$ and all subsequent packets in advance. Since packet $k$ is still submitted when the queue is empty and the edge server is idle in policy $\pi'$, the waiting time of these packets remains unchanged. Therefore, we have $W^{'}_j=W_j$ for all $j\ge k$. Then we have $P^{'}_k=P_k-G_k$ and $P^{'}_j=P_j$ for $j> k$. 

According to the definition of the average PAoI in \eqref{def:PAoI}, policy $\pi'$ achieves a lower average PAoI than policy $\pi$. This contradicts the assumption that $\pi$ is optimal. Therefore, the optimal policy must satisfy $S_k\le D_{k-1}$ for all $k=1,2,\cdots$ in the optimal policy.


\subsection{Proof of Theorem \ref{The:WOP_RT}}\label{proof:The:WOP_RT}

We prove this theorem by showing that any continuous working policy can be outperformed by a randomized threshold policy. Our approach is inspired by \cite{arafa2020agemin}, but we make adjustments and adaptations to their method to account for the different optimization objectives and models we consider.

Under a continuous working policy, the action of packet $k$ $Z_k=T_k+W_k+\min\{\Xi_k,C_k\}$, where $\Xi_k$ is a random value determined by all history information $\{S_0,T_0,C_0,W_0,\cdots,S_{k-1},T_{k-1},C_{k-1},W_{k-1}\}$ before packet $k$, denoted by $\mathcal{H}_k$, packet $k$'s transmission time duration $T_k$ and packet $k$'s waiting time $W_k$. Specifically, the source first observes $\mathcal{H}_k$, $T_k$ and $W_k$, and then chooses $\Xi_k$ by a conditional probability $p(\gamma,\eta , \pmb{h})\triangleq Pr(\Xi_k|T_k=\gamma,W_k=\eta,\mathcal{H}_k=\pmb{h})$. By the definition of $W_k$ we have
\begin{align}
  W_k=&\max\{0,D_{k-1}-S_{k}-T_k\}\\
  =&\max\{0,C_{k-1}-\Xi_{k-1}-T_k\}.\label{equ:Wk}
\end{align}
Since $T_k$'s and $C_k$'s are independent of the scheduling policy, a policy can be equivalently written as $\pi =\{\Xi_k,i\ge 1\}$. The $k+1$-th peak AoI is
\begin{align}
  P_{k+1}&=X_{k+1}+T_{k+1}+W_{k+1}+C_{k+1}\\
  &=Z_k+T_{k+1}+W_{k+1}+C_{k+1}\\
  &=T_k+W_k+\min\{\Xi_k,C_k\}+T_{k+1}+W_{k+1}+C_{k+1}.\label{equ:Pk1}
\end{align}
To facilitate the analysis, we define a special value for packet $k$ as
\begin{align}
  R_k\triangleq\min\{\Xi_k,C_k\}+2W_{k+1}.\label{def:Rk}
\end{align}
Based on \eqref{equ:Wk}, $T_{k+1}$ also affects $R_k$. For packet $k$ and a fixed history $\mathcal{H}_k$, we group all the status updating sample paths that have the same $T_k,W_{k}$ and perform a statistical averaging over all of them to get the following average $R_k$:
\begin{equation}
  \widehat{R}_{k}(\gamma,\eta,\pmb{h}) \triangleq \mathbb{E}\left[R_k|T_k=\gamma,W_{k}=\eta,\mathcal{H}_k=\pmb{h}\right],\label{def:wRk}
\end{equation}
where the expectation is taken for $\Xi_{k},C_k$ and $T_{k+1}$.

Then, we derive a lower bound of the average PAoI as shown in \eqref{equ:a}-\eqref{equ:g}. There, \eqref{equ:a} follows by \eqref{def:PAoI}; \eqref{equ:b} follows by \eqref{equ:Pk1}; \eqref{equ:c} follows by combining the random variables with the same subscript together and the fact that $T_k$'s and $C_k$'s are i.i.d.; \eqref{equ:d} follows that: by \eqref{equ:Wk} we have $T_k+W_k\le C_k$, and then we have 
\begin{align*}
  &\lim_{N \to \infty}\frac{\mathbb{E}\left[T_0+W_0-T_N-W_N\right]}{N}\\
  \le& \lim_{N \to \infty}\frac{\mathbb{E}\left[T_0+W_0\right]}{N}
  \le \lim_{N \to \infty}\frac{\mathbb{E}\left[C_0\right]}{N}=0,
\end{align*}
and
\begin{align*}
  &\lim_{N \to \infty}\frac{\mathbb{E}\left[T_0+W_0-T_N-W_N\right]}{N}\\
  \ge& -\lim_{N \to \infty}\frac{\mathbb{E}\left[T_N+W_N\right]}{N}
  \ge -\lim_{N \to \infty}\frac{\mathbb{E}\left[C_N\right]}{N}=0,
\end{align*}
; \eqref{equ:e} follows by \eqref{def:Rk}; $R^*(\mathcal{H}_k)$ in \eqref{equ:f} denotes the minimum value of $\widehat{R}_{k}(\gamma,\eta,h)$ over all possible $T_k$ and $W_{k}$; $R_{\min}$ in \eqref{equ:g} denotes the minimum value of $R^*(\mathcal{H}_k)$ over all packets and their corresponding histories, i.e, the minimum over all $k$ and $\mathcal{H}_k$.
\newline Note that in the continuous working policy achieving $R_{\min}$, $\Xi_k$ is determined by a certain distribution $p(\gamma,\eta,\pmb{h})= Pr(\Xi_k|T_k=\gamma,W_k=\eta,\mathcal{H}_k=\pmb{h})$, for a fixed condition, i.e., fixed values of $T_k$, $W_k$, and $\mathcal{H}_k$. Now observe that $T_i$'s are i.i.d.. Therefore, if we directly apply the distribution that achieves $R_{\min}$ over all packets without considering the history information, all inequations in \eqref{equ:a}-\eqref{equ:g} become equations. The new policy achieves the lower bound and is a fixed threshold policy. This completes the proof.

\begin{figure*}[!t]
  \normalsize
  \begin{align}
    &\lim_{N \to \infty}\frac{1}{N}\mathbb{E}_{\pi}\left[\sum_{k=0}^{N-1}P_{k+1}\right]=\lim_{N \to \infty}\frac{\sum_{k=0}^{N-1}\mathbb{E}\left[P_{k+1}\right]}{N}\label{equ:a}\\
    =&\lim_{N \to \infty}\frac{\sum_{k=0}^{N-1}\mathbb{E}\left[T_k+W_k+\min\{\Xi_k,C_k\}+T_{k+1}+W_{k+1}+C_{k+1}\right]}{N}\label{equ:b}\\
    =&\lim_{N \to \infty}\frac{\sum_{k=0}^{N-1}\mathbb{E}\left[\min\{\Xi_k,C_k\}+2W_{k+1}\right]}{N}+2\mathbb{E}\left[T\right]+\mathbb{E}\left[C\right]+\lim_{N \to \infty}\frac{\mathbb{E}\left[T_0+W_0-T_N-W_N\right]}{N}\label{equ:c}\\
    =&\lim_{N \to \infty}\frac{\sum_{k=0}^{N-1}\mathbb{E}\left[\min\{\Xi_k,C_k\}+2W_{k+1}\right]}{N}+2\mathbb{E}\left[T\right]+\mathbb{E}\left[C\right]\label{equ:d}\\
    =&\lim_{N \to \infty}\frac{\sum_{k=0}^{N-1}\mathbb{E}_{\mathcal{H}_k}\left[\mathbb{E}_{T_k}\left[\mathbb{E}_{W_{k}}\left[\widehat{R}_k(T_k,W_k,\mathcal{H}_k)|W_k\right]|T_k\right]|\mathcal{H}_k\right]}{N}+2\mathbb{E}\left[T\right]+\mathbb{E}\left[C\right]\label{equ:e}\\
    \ge &\lim_{N \to \infty}\frac{\sum_{k=0}^{N-1}\mathbb{E}_{\mathcal{H}_k}\left[R^*(\mathcal{H}_k)|\mathcal{H}_k\right]}{N}+2\mathbb{E}\left[T\right]+\mathbb{E}\left[C\right]\label{equ:f}\\
    \ge &\lim_{N \to \infty}\frac{\sum_{k=0}^{N-1}R_{\min}}{N}+2\mathbb{E}\left[T\right]+\mathbb{E}\left[C\right]\\
    \ge &R_{\min}+2\mathbb{E}\left[T\right]+\mathbb{E}\left[C\right]\label{equ:g}
  \end{align}  
  \hrulefill
  \vspace*{2pt}
\end{figure*}

\subsection{Proof of Theorem \ref{The:WOP_FT}}\label{proof:The:WOP_FT}
For any randomized threshold policy, the average of PAoI as shown in \eqref{def:prob2} is
\begin{align*}
  \mathbb{E}[P]=&2\mathbb{E}\left[T\right]+\mathbb{E}\left[C\right]\\
  &+\int_{0}^{\infty} f_{\Theta}(\theta) \bigg( \theta \int_{\theta}^{\infty}f_C(x)\,dx+ \int_{0}^{\theta}xf_C(x)  \,dx  \\
  & +2\int_{0}^{\infty}f_T(x)\int_{x+\theta}^{\infty} f_C(y)(y-\theta-x) \,dy   \,dx \bigg) \,d\theta\\
  \ge &2\mathbb{E}\left[T\right]+\mathbb{E}\left[C\right]\\
  &+\mathop{\min}_{\theta}\bigg( \theta \int_{\theta}^{\infty}f_C(x)\,dx+ \int_{0}^{\theta}xf_C(x)  \,dx  \\
  & +2\int_{0}^{\infty}f_T(x)\int_{x+\theta}^{\infty} f_C(y)(y-\theta-x) \,dy   \,dx\bigg)
\end{align*}
Suppose $\theta^*$ achieves the minimum of the above inequality. Then, construct a new policy $\pi'$ satisfying $Pr(\Theta'=\theta^*)=1$. Note that policy $\pi^*$ is a fixed threshold policy defined in Sec.~\ref{subdec:FT}. Therefore, any randomized threshold policy is outperformed by a fixed threshold policy. This completes the proof.

\subsection{Proof of Lemma \ref{Lem:expPz}}\label{proof:Lem:expPz}
Based on \eqref{equ:EC} and \eqref{equ:EW} we have 
\begin{align}
  P''(\theta)=2\int_{0}^{\infty}f_C(x+\theta)\lambda e^{-\lambda x}\,dx-f_C(\theta)\label{proof_Pz}
\end{align}
For any $\theta$ that satisfies \eqref{equ_fi}, by submitting $f_T(x)=\lambda e^{-\lambda x}, x > 0$ into \eqref{equ_fi}, we have
\begin{align*}
  F_C(\theta)+1&=2\int_{0}^{\infty}f_T(x)F_C(x+\theta)\,dx\\
  &=2\int_{\theta}^{\infty}f_C(y)\int_{0}^{y-\theta} f_T(x) \,dx\,dy\\
  &=2\int_{\theta}^{\infty}f_C(y)(1-e^{-\lambda(y-\theta)})\,dy\\
  &=2-2\int_{\theta}^{\infty}f_C(y)e^{-\lambda(y-\theta)}\,dy\\
  &=2-2\int_{0}^{\infty}f_C(x+\theta)e^{-\lambda x}\,dx
\end{align*}
Finally, we can obtain $P''(\theta)=\lambda(1-F_C(\theta))-f_C(\theta)$ by submitting the above last equation into \eqref{proof_Pz}.

\subsection{Proof of Proposition~\ref{Lem:pz}}\label{proof:lempz}
We prove the proposition by contradiction. In the first case, if $P''_z(\theta)>0$ for all $\theta_1<\theta<\theta_2$, we suppose that $\theta_a,\theta_b$ ($\theta_1<\theta_a<\theta_b<\theta_2$) satisfy \eqref{equ_fi}, and there is no other $\theta \in (\theta_a,\theta_b)$ that satisfies \eqref{equ_fi}. Since $P'(\theta)$ is continuous, there must exist $\theta_c\in (\theta_a,\theta_b)$ such that $P'(\theta_c)=0$. Based on our assumption, we have $P''(\theta_c)<0$. Then, we have $P''_z(\theta_c)=P''(\theta_c)<0$, which contradicts $P''_z(\theta)>0$ for all $\theta_1<\theta<\theta_2$. Thus, if $P''_z(\theta)>0$ for all $\theta_1<\theta<\theta_2$, there is at most one $\theta\in (\theta_1,\theta_2)$ that satisfies \eqref{equ_fi}. If $\theta\in (\theta_1,\theta_2)$ satisfies \eqref{equ_fi}, we have $P''(\theta)=P''_z(\theta)>0$, which means $\theta$ also satisfies \eqref{equ_fi2}. We complete the proof of the first case.

In the second case, if $P''_z(\theta)>0$ for all $\theta_1<\theta<\theta_2$, we suppose that $\theta_d$ ($\theta_1<\theta_d<\theta_2$) satisfies \eqref{equ_fi} and \eqref{equ_fi2}. Then, we can obtain $P''_z(\theta_d)=P''(\theta_d)>0$, which also contradicts $P''_z(\theta)>0$ for all $\theta_1<\theta<\theta_2$.  Therefore, if $P''_z  (\theta)<0$ for all $\theta_1<\theta<\theta_2$, there is no $\theta\in [\theta_1,\theta_2]$ that satisfies \eqref{equ_fi}. We complete the proof of the second case.

\subsection{Proof of Theorem \ref{The:WP_SD}}\label{proof:The:WP_SD}
We solve this proof in two steps: First, we use similar method in Proof of Theorem~\ref{The:WOP_RT} to prove that the optimal policy is stationary. Then, we turn Problem~\eqref{def:prob} under stationary policies to an infinite-horizon MDP in one \textit{epoch} and prove the optimal policy is deterministic.

\textbf{Step 1:} Similar with Proof of Theorem~\ref{The:WOP_RT}, we show any continuous working policy is outperformed by a stationary policy. Under a continuous working policy, $Z_i=T_i+\min\{\Xi_i,C_i\}$, where $\Xi_i$ is a random value determined by all history information $\{S_0,T_0,C_0,\cdots,S_{i-1},T_{i-1},C_{i-1}\}$ before packet $i$, denoted by $\mathcal{H}_i$ and packet $i$'s transmission time duration $T_i$. Specifically, the source first observes $\mathcal{H}_i$ and $T_i$, and then chooses $\Xi_i$ by a conditional probability $p(\gamma, \pmb{h})\triangleq Pr(\Xi_i|T_i=\gamma,\mathcal{H}_k=\pmb{h})$. We use $\Omega_i$ to denote the event packet $i$ is delivered successfully, which is given by $\Omega_i: C_i\le \Xi_i+T_{i+1}$. Define
\begin{equation}
  R_i(\Xi_i)=Z_i+(T_i+C_i)\mathbf{1}_{\Omega_i}.\label{def:Ri}
\end{equation}
Note that $\Omega_i$ is determined by $C_i,\Xi_i$ and $T_{i+1}$. For a fixed history $\mathcal{H}_i$, we group all the status updating sample paths that have the same $T_i$ and perform a statistical averaging over all of them to get the following average $R_i(\Xi_i)$ for packet $i$:
\begin{equation}
  \widehat{R}_i(\gamma,\mathcal{H}_i) \triangleq \mathbb{E}\left[R_i(\Xi_i)|T_i=\gamma,\mathcal{H}_i\right],\label{def:hatR}
\end{equation}
where the expectation is taken for $\Xi_i,C_i$ and $T_{i+1}$. Similarly, define the average $\mathbf{1}_{\Omega_i}$ as
\begin{equation}
  \widehat{x}_i(\gamma,\mathcal{H}_i) \triangleq \mathbb{E}\left[\mathbf{1}_{\Omega_i}|T_i=\gamma,\mathcal{H}_i\right].\label{def:hatx}
\end{equation}
Without loss of generality, suppose packet $N$ is the $K$-th successfully delivered packet. Then we have 
\begin{align}
  &\lim_{K \to \infty}\frac{1}{K}\mathbb{E}_{\pi}\left[\sum_{k=1}^{K}P_{k}\right]\label{equ:3a}\\
  =&\lim_{N \to \infty}\frac{\mathbb{E}_{\pi}\left[\sum_{k=1}^{N}\sum_{j=i_{k-1}}^{i_k-1}(Z_j+(T_{j+1}+C_{j+1})\mathbf{1}_{\Omega_{j+1}})\right]}{\mathbb{E}_{\pi}\left[\sum_{i=1}^{N}\mathbf{1}_{\Omega_{i}}\right]}\label{equ:3b}\\
  =&\lim_{N \to \infty}\frac{\mathbb{E}_{\pi}\left[\sum_{i=0}^{N-1}Z_{i}+\sum_{i=1}^{N}\left((T_{i}+C_{i})\mathbf{1}_{\Omega_{i}}\right)\right]}{\mathbb{E}_{\pi}\left[\sum_{i=1}^{N}\mathbf{1}_{\Omega_{i}}\right]}\label{equ:3c}\\
  =&\lim_{N \to \infty}\frac{\mathbb{E}_{\pi}\left[\sum_{i=1}^{N}\left(Z_{i}+(T_{i}+C_{i})\mathbf{1}_{\Omega_{i}}\right)\right]}{\mathbb{E}_{\pi}\left[\sum_{i=1}^{N}\mathbf{1}_{\Omega_{i}}\right]}.\label{equ:3d}
\end{align}
In the above equations, \eqref{equ:3b} follows by \eqref{equ:Pk_}; \eqref{equ:3c} follows by equivalent transformations; and \eqref{equ:3d} follows that, $\lim_{N \to \infty}\frac{Z_N-Z_0}{\mathbb{E}_{\pi}\left[\sum_{i=1}^{N}\mathbf{1}_{\Omega_{i}}\right]}=0$. Then, we derive a lower bound of the average PAoI achieving by a continuous working policy as follows.
\begin{align}
  &\lim_{N \to \infty}\frac{\mathbb{E}_{\pi}\left[\sum_{i=1}^{N}\left(Z_{i}+(T_{i}+C_{i})\mathbf{1}_{\Omega_{i}}\right)\right]}{\mathbb{E}_{\pi}\left[\sum_{i=1}^{N}\mathbf{1}_{\Omega_{i}}\right]}\label{equ:2a}\\
  =&\frac{\sum_{i=1}^{\infty}\mathbb{E}\left[R_i(\Xi_i)\right]}{\sum_{i=1}^{\infty}\mathbb{E}\left[\mathbf{1}_{\Omega_i}\right]}\label{equ:2b}\\
  =&\frac{\sum_{i=1}^{\infty}\mathbb{E}_{\mathcal{H}_i}\left[\mathbb{E}_{T_i}\left[\widehat{R}_i(T_i,\mathcal{H}_i)|T_i\right]|\mathcal{H}_i\right]}{\sum_{i=1}^{\infty}\mathbb{E}\left[\mathbf{1}_{\Omega_i}\right]}\label{equ:2c}\\
  =&\frac{\sum_{i=1}^{\infty}\mathbb{E}_{\mathcal{H}_i}\left[\mathbb{E}_{T_i}\left[\widehat{x}_i(T_i,\mathcal{H}_i)|T_i=\gamma\right]\frac{\mathbb{E}_{T_i}\left[\widehat{R}_i(T_i,\mathcal{H}_i)|T_i\right]}{\mathbb{E}_{T_i}\left[\widehat{x}_i(T_i,\mathcal{H}_i)|T_i\right]}|\mathcal{H}_i\right]}{\sum_{i=1}^{\infty}\mathbb{E}\left[\mathbf{1}_{\Omega_i}\right]}\label{equ:2d}\\
  \ge &\frac{\sum_{i=1}^{\infty}\mathbb{E}_{\mathcal{H}_i}\left[\mathbb{E}_{T_i}\left[\widehat{x}_i(T_i,\mathcal{H}_i)|T_i\right]R^*(\mathcal{H}_i)|\mathcal{H}_i\right]}{\sum_{i=1}^{\infty}\mathbb{E}\left[\mathbf{1}_{\Omega_i}\right]}\label{equ:2e}\\
  \ge &\frac{\sum_{i=1}^{\infty}\mathbb{E}_{\mathcal{H}_i}\left[\mathbb{E}_{T_i}\left[\widehat{x}_i(T_i,\mathcal{H}_i)|T_i\right]|\mathcal{H}_i\right]R_{\min}}{\sum_{i=1}^{\infty}\mathbb{E}\left[\mathbf{1}_{\Omega_i}\right]}\label{equ:2f}\\
  =&R_{\min}.
\end{align}
In the above equations and inequalities, \eqref{equ:2b} follows by the monotone convergence theorem; \eqref{equ:2c} follows by \eqref{def:hatR}; $R^*(\mathcal{H}_i)$ in \eqref{equ:2d} denotes the minimum value of $\frac{\mathbb{E}_{T_i}\left[\widehat{R}_i(T_i,\mathcal{H}_i)|T_i\right]}{\mathbb{E}_{T_i}\left[\widehat{x}_i(T_i,\mathcal{H}_i)|T_i\right]}$; $R_{\min}$ in \eqref{equ:2f} denotes the minimum value of $R^*(\mathcal{H}_i)$ over all packets and their corresponding histories, i.e., the minimum over all $i$ and $\mathcal{H}_i$; and \eqref{equ:2f} follows by the relationships between $\widehat{x} _i$ and $x_i$ which are the same as those between $\widehat{R} _i$ and $R_i$  \eqref{def:hatx}. These equations give a lower bound for the average PAoI under the continuous working policy.
\newline Note that in the continuous working policy achieving $R_{\min}$, $\Xi_i$ is determined by a conditional distribution $Pr(\Xi_i|T_i)=p(\gamma,\pmb{h})=Pr(\Xi_i|T_i,\mathcal{H}_k=\pmb{h})$, which is only depends on $T_i$. Now observe that $T_i$'s are i.i.d.. Therefore, if we directly apply the distribution that achieves $R_{\min}$ over all packets without considering the history information, all inequations in \eqref{equ:2a}-\eqref{equ:2f} become equations. The new policy achieves the lower bound and is a stationary policy.

\textbf{Step 2:} We use the term \textit{epoch} to denote the time between two consecutive successfully updates. For example, the $k$-th epoch starts at time $D_{i_{k-1}}$ and ends at time $D_{i_{k}}$. Each epoch has a corresponding PAoI value. Since the optimal policy is stationary and $T_k$'s and $C_k$'s are i.i.d., PAoI values from different epochs are independent. Therefore, the minimum average PAoI achieved in one epoch is also the optimal result for \eqref{def:prob} in the case $\omega\equiv \text{wp}$. We redefine some notations to fit in the epoch definition, let (1) $T_r^e$ denote transmission time of the $r$-th updates in one epoch; (2) $C^e_r$ denote computation time of the $r$-th updates; (3) $Z^e_r=T^e_r+\min\{\Xi^e_r,C^e_r\}$ denote the time between the transmission instance of the $r$-th and $r+1$-th updates; and (4) $\Omega_r^e:C_r^e\le \Xi_r^e+T^e_{r+1}$.

Suppose there are total $R$ updates in one epoch, the PAoI value is equivalent to $P^e=\sum_{r=0}^{R-1}Z^e_r+T^e_R+C_R^e$. We consider the following problem:
\begin{align}
  \min_{\Xi^e_0,\cdots}&\mathbb{E}\left[\sum_{r=0}^{R-1}Z^e_r+T^e_R+C^e_R\right]\label{def:prob_proof}.
\end{align}
Let $\mathbf{I}_{r}^R$ denote event $R\ge r$, which is given by
\begin{align*}
  \mathbf{I}_{r}^R=\prod_{i=1}^{r-1}\mathbf{1}_{\overline{\Omega^e_i}}.
\end{align*}
Note that we consider the $0$-th update, the start of one epoch, is a successfully delivered packet, thus, at least one packet in each epoch. Using this, the average PAoI in one epoch is equivalent to
\begin{align*}
  &\mathbb{E}\left[\sum_{r=1}^{R}Z^e_{r-1}+T^e_R+C^e_R\right]=\mathbb{E}\left[\sum_{r=1}^{\infty}Z^e_{r-1}\mathbf{I}_{r}^{R}+T^e_{R}+C^e_R\right]\\
  =&\mathbb{E}\left[\sum_{r=1}^{\infty}\prod_{i=1}^{r-1}\mathbf{1}_{\overline{\Omega^e_i}}(Z^e_{r-1}+(T^e_r+C^e_r)\mathbf{1}_{\Omega^e_r})\right]
\end{align*}
To solve Problem~\eqref{def:prob_proof}, we formulate an infinite-horizon continuous-state MDP problem with the following elements:
\begin{itemize}
  \item State $\mathcal{S}_r$: the transmission time of the $(r-1)$-th update packet $T^e_{r-1}$.
  \item Action $\mathcal{A}_r$: the parameter $\Xi^e_{r-1}$ determined according to a conditional probability.
  \item Cost: 
  \begin{align*}
    c(\mathcal{S}_r,\mathcal{A}_r)&=\mathbb{E}\left[\prod _{i=1}^{r-1}\mathbf{1}_{\overline{\Omega^e_i}}(Z^e_{r-1}+(T^e_r+C^e_r)\mathbf{1}_{\Omega^e_r})\right]\\
    &=P(\bigcap  _{i=1}^{r-1}\overline{\Omega^e_i})\mathbb{E}\left[Z^e_{r-1}+(T^e_r+C^e_r)\mathbf{1}_{\Omega^e_r}\right].
  \end{align*}
  \item State transition probabilities:
  \begin{align*}
    P(T^e_{r+1}=y|T^e_{r}=x)=P(T^e_{r+1}=y).
  \end{align*}
\end{itemize}
The cost-to-go function is defined as
\begin{align*}
  J_{k}(x)=\min_{\mathcal{A}_k,\cdots}\mathbb{E}\left[\sum_{r=k}^{\infty}c(\mathcal{S}_r,\mathcal{A}_r)|\mathcal{S}_k=x\right]
\end{align*}
Define the cost-to-go function under an arbitrary stationary policy $\pi$ as
\begin{align*}
  J_{\pi,k}(x)=\mathbb{E}_{\pi}\left[\sum_{r=k}^{\infty}c(\mathcal{S}_r,\mathcal{A}_r)|\mathcal{S}_k=x\right]
\end{align*}
Since $\mathbb{E}\left[Z^e_{r-1}+(T^e_r+C^e_r)\mathbf{1}_{\Omega^e_r}\right]\le 2\mathbb{E}\left[T^e_r+C^e_r\right]<\infty$ and $P(\bigcap _{i=1}^{r-1}\overline{\Omega^e_i})\rightarrow 0$ as $r\rightarrow \infty$, for a sufficiently large $M$, we have $J_{\pi,M+1}\thickapprox 0$. Based on the backward recursion of the stochastic Bellman's dynamic programming, $J_{\pi,k}(x)$ can be expressed recursively as
\begin{align*}
  J_{\pi,k}(x)=\mathbb{E}_{\pi}\left[c(x,\mathcal{A}_k)+\int_{0}^{\infty}J_{k+1}(y)f_{T}(y)  \,dy \right],
\end{align*}
Denote $\pi^*$ as the optimal stationary policy. Note that $J_{\pi^*,k}(x)$ is the optimal cost-to-go function $J_k$. Clearly, $J_{\pi,M+1}(x)=J_{\pi^*,M+1}(x)\approx 0$. For $k\le M$, we have
\begin{align*}
  J_{\pi,k}(x)=&\int_{x}^{\infty} \left[c(x,u)+\int_{0}^{\infty}J_{\pi,k+1}(y)f_{T}(y)\,dy\right]f_{\mathcal{A}_k}(u)\,du\\
  \ge&\int_{x}^{\infty} \left[c(x,u)+\int_{0}^{\infty}J_{\pi^*,k+1}(y)f_{T}(y)\,dy\right]f_{\mathcal{A}_k}(u)\,du\\
  \ge&\min_{\Xi^e_{k-1}}\left[c(x,\Xi^e_{k-1})+\int_{0}^{\infty}J_{\pi^*,k+1}(y)f_{T}(y)\,dy\right]\\
  =&J_{\pi^*,k}
\end{align*}
The last inequality follows since a convex combination is always larger than the minimum. By choosing a policy $\hat{\pi}$ with $P(\Xi^e_{k-1}=u|T^e_{k-1}=x)=1$ if
\begin{align*}
  u=\min_{\Xi_{k-1}^e}\left[c(x,\Xi^e_{k-1})+\int_{0}^{\infty}J_{\pi^*,k+1}(y)f_{T}(y)\,dy\right],
\end{align*}
all the inequalities above become equalities. Note that in $\hat{\pi}$, $\Xi^e_r$ is a deterministic function about $T^e_{r}$, which mean that $\hat{\pi}$ is a stationary deterministic policy and achieves the optimal solution for Problem~\eqref{def:prob_proof}. This completes the proof.

\subsection{Proof of Lemma~\ref{Lem:pc}}\label{proof:Lem:pc}
Let $c_1>0$, and let the solution of Problem~\eqref{def:probX_const_lag} be given by $\theta^{(1)}$ for $c=c_1$. Now for some $c_2>c_1$, we have
  \begin{align*}
    p(c_1)&=\min_{\theta} \mathbb{E}\left[T+\min\{\theta,C\}+(T+C)\mathbf{1}_{\Omega}\right]-c_1 Pr(\Omega)\\
    &\ge \min_{\theta} \mathbb{E}\left[T+\min\{\theta,C\}+(T+C)\mathbf{1}_{\Omega}\right]-c_2 Pr(\Omega)\\
    &\ge p(c_2),
  \end{align*}
  where the last inequality follows since $\theta^{(1)}$ is also feasible in Problem~\eqref{def:probX_const_lag} for $c=c_2$.

  Next, note that both Problem~\eqref{def:probX_const_lag} and Eq.~\ref{def:probX_const} have the same feasible set. In addition, if $p(c)=0$, then the objective function of \eqref{def:probX_const} satisfies
  \begin{align*}
    \frac{\mathbb{E}\left[T+\min\{\theta,C\}+(T+C)\mathbf{1}_{\Omega}\right]}{Pr(\Omega)}=c.
  \end{align*}
  Hence, the objective function of \eqref{def:probX_const} is minimized by finding the minimum $c\ge 0$ such that $p(c)=0$. Finally, by the first part of lemma, there can only be one such $c$, which we denote $c^*$.

\subsection{Proof of Theorem~\ref{The:optiamlC}}\label{proof:optimalC}
Similar with Sec.~\ref{subsec:fixed}, in order to solve the Problem~\eqref{def:probX}, we consider the following functional parameterized problem:
\begin{align}
  f(c)\triangleq \min_{g:\theta=g(T)} &\mathbb{E}\left[Z+(T+C)\mathbf{1}_{\Omega}\right]-cPr(\Omega)\label{def:probX_const_lag_p}\\
   s.t. \quad &g(T)\ge 0
\end{align}
There exist a $c^*$ such that $f(c^*)=0$ and the optimal policy for Problem~\eqref{def:probX} is the solution for this parameterized problem with $c=c^*$. 

Then, we use the Lagrangian duality approach to solve the above parameterized problem. Since the computation time is exponentially distributed, we have $f_C(x)=\mu e^{-\mu x}$. By submitting $f_C(x)$ into \eqref{equ:EC}\eqref{equ:PO} and \eqref{equ:ETC}, we have 
\begin{align*}
  \mathbb{E}\left[\min\{\theta,C\}\right]&=\frac{1}{\mu}-\frac{1}{\mu}\int_{0}^{\infty}f_T(x)e^{-\mu g(x)}\,dx\\
  Pr(\Omega)&=1-\mathcal{L}_{\mu}\int_{0}^{\infty}f_T(x)e^{-\mu g(x)}\,dx\\
  \mathbb{E}\left[T \mathbf{1}_{\Omega}\right]&=\mathbb{E}[T]-\mathcal{L}_{\mu}\int_{0}^{\infty}f_T(x)xe^{-\mu g(x)}\,dx\\
  \mathbb{E}\left[C \mathbf{1}_{\Omega}\right]&=\frac{1}{\mu}-\frac{1}{\mu}\int_{0}^{\infty}f_T(x)e^{-\mu g(x)}\,dx\\
  &\quad -\mathcal{L}_{\mu}\int_{0}^{\infty}f_T(x)g(x)e^{-\mu g(x)}\,dx\\
  &\quad -\mathcal{M}_{\mu}\int_{0}^{\infty}f_T(x)e^{-\mu g(x)}\,dx
\end{align*}
where $\mathcal{L}_{\mu}=\int_{0}^{\infty}f_T(x)e^{-\mu x}\,dx$ is the Laplace transform of the transmission time distribution and $\mathcal{M}_{\mu}=\int_{0}^{\infty}xf_T(x)e^{-\mu x}\,dx$ is the Laplace transform of $xf_T(x)$. The Lagrangian is
\begin{align*}
  &L=\frac{2}{\mu}+2\mathbb{E}[T]-c-\int_{0}^{\infty}u(x)g(x)\,dx \\
  &+\int_{0}^{\infty}f_T(x)e^{-\mu g(x)}(c\mathcal{L}_{\mu}-\mathcal{M}_{\mu}-\frac{2}{\mu}-\mathcal{L}_{\mu}(g(x)+x))\,dx
\end{align*}
with additional complementary slackness condition
\begin{align}
  u(x)g(x)=0, \quad \forall x \ge 0.\label{csc}
\end{align}
We apply the KKT optimally conditions to Lagrangian function. Taking derivative with respect to $g(x)$ and equating to $0$ we get
\begin{align}
  g(x)+x=\frac{\frac{u(x)}{f(x)e^{-\mu g(x)}}+\mathcal{L}_{\mu}+\mu c\mathcal{L}_{\mu}-\mu\mathcal{M}_{\mu}-2}{\mu\mathcal{L}_{\mu}}\label{equ:gxx}
\end{align}
For notational simplicity, let us define
$\gamma =\frac{\mathcal{L}_{\mu}+\mu c\mathcal{L}_{\mu}-\mu\mathcal{M}_{\mu}-2}{\mu\mathcal{L}_{\mu}}$. The optimal solution $g(x)$ is obtained by considering the following three cases:

\textit{Case 1}: If $\gamma \le 0$, then by \eqref{csc} and $g(x)\ge 0$, we obtain $g(x)=0$.

\textit{Case 2}: If $\gamma > 0$ and $u(x)=0$, then by \eqref{equ:gxx}, we obtain $g(x)=\gamma -x$. In this case,  we require $\gamma -x\ge 0$ since $g(x)\ge 0$.

\textit{Case 3}: If $\gamma > 0$ and $u(x)>0$, then by \eqref{csc}, we obtain $g(x)=0$.

In summary, the optimal solution for Problem~\eqref{def:probX_const_lag_p} is $g(x)=\max\{0,\gamma -x\}$, which is a transmission-aware threshold policy. Therefore, when $c=c^*$, the solution Problem~\eqref{def:probX_const_lag_p} gives the optimal policy for Problem~\eqref{def:probX}. Since $f(c^*)=0$, we have
\begin{align*}
  c^*=\frac{\mathbb{E}\left[Z+(T+C)\mathbf{1}_{\Omega}\right]}{Pr(\Omega)},
\end{align*}
which means that $c^*$ is the minimum average PAoI $ \mathcal{P}^{*,\text{wp}}$. This completes the proof.

 
%

\bibliographystyle{IEEEtran}

\bibliography{IEEEabrv,refer}

\vfill

\end{document}